\documentclass[11pt]{article}

\usepackage{amsfonts}
\usepackage{amssymb}
\usepackage{amstext}
\usepackage{amsmath}
\usepackage{mathtools}

\usepackage{amsthm} 

\usepackage{caption}
\usepackage{subcaption}

\usepackage{color}
\usepackage{nameref}
\definecolor{ForestGreen}{rgb}{0.1333,0.5451,0.1333}
\definecolor{DarkRed}{rgb}{0.8,0,0}
\definecolor{Red}{rgb}{0.9,0,0}
\usepackage[linktocpage=true,
	pagebackref=true,colorlinks,
	linkcolor=DarkRed,citecolor=ForestGreen,
	bookmarks,bookmarksopen,bookmarksnumbered]
	{hyperref}
\usepackage{cleveref}

\usepackage{thmtools,thm-restate} 

\usepackage[numbers,sort&compress]{natbib}

\usepackage{graphicx}
\usepackage{graphics}
\usepackage{colordvi}
\usepackage{xspace}
\usepackage{algorithm}
\usepackage{algorithmicx}
\usepackage{url}
\usepackage{enumitem}

        {\hspace*{\fill}$\Box$\par\vspace{4mm}}

\usepackage[left=1in,top=1in,right=1in,bottom=1in]{geometry} 
\usepackage{booktabs}
\usepackage{threeparttable}
\usepackage{times} 

\makeatletter
\renewcommand{\paragraph}{%
  \@startsection{paragraph}{4}%
  {\z@}{1ex \@plus 1ex \@minus .2ex}{-1em}%
  {\normalfont\normalsize\bfseries}%
}
\makeatother

\makeatletter
\def\thmt@refnamewithcomma #1#2#3,#4,#5\@nil{%
  \@xa\def\csname\thmt@envname #1utorefname\endcsname{#3}%
  \ifcsname #2refname\endcsname
    \csname #2refname\expandafter\endcsname\expandafter{\thmt@envname}{#3}{#4}%
  \fi
}
\makeatother

\declaretheorem[numberwithin=section,refname={Theorem,Theorems},Refname={Theorem,Theorems}]{theorem}
\declaretheorem[numberlike=theorem,refname={Lemma,Lemmas},Refname={Lemma,Lemmas}]{lemma}

\declaretheorem[numberlike=theorem,refname={Claim, Claims},Refname={Claim, Claims}]{claim}

\declaretheorem[numberlike=theorem]{definition}

\newtheorem{invariant}[theorem]{Invariant}
\renewcommand{\phi}{\varphi}

\renewcommand{\t}{\text{slack}}
\newcommand{\s}{\text{status}}
\newcommand{\N}{\mathcal{N}}
\newcommand{\B}{\mathcal{B}}
\newcommand{\C}{\mathcal{C}}
\renewcommand{\P}{\mathcal{P}}
\newcommand{\eps}{\epsilon}
\newcommand{\dl}{\delta}
\newcommand{\dd}{D}
\newcommand{\G}{\mathcal{G}}
\newcommand{\V}{\mathcal{V}}
\newcommand{\E}{\mathcal{E}}

\ifdefined\ShowComment

\def\danupon#1{\marginpar{$\leftarrow$\fbox{D}}\footnote{$\Rightarrow$~{\sf #1 --Danupon}}}
\def\babis#1{\marginpar{$\leftarrow$\fbox{B}}\footnote{$\Rightarrow$~{\sf #1 --Babis}}}
\def\sayan#1{\marginpar{$\leftarrow$\fbox{S}}\footnote{$\Rightarrow$~{\sf #1 --Sayan}}}
\def\monika#1{\marginpar{$\leftarrow$\fbox{M}}\footnote{$\Rightarrow$~{\sf #1 --Monika}}}

\else

\def\danupon#1{}
\def\babis#1{}
\def\sayan#1{}
\def\monika#1{}

\fi

\newboolean{short}
\setboolean{short}{true} 

\newcommand{\shortOnly}[1]{\ifthenelse{\boolean{short}}{#1}{}}
\newcommand{\longOnly}[1]{\ifthenelse{\boolean{short}}{}{#1}}

\title{Deterministic Fully Dynamic Data Structures for Vertex Cover and Matching\thanks{An extended abstract of this paper, not containing the algorithm in Section~\ref{sec:matching:sqrtn}  and not containing the proof of Theorem~\ref{th:vc:runtime}, has been accepted for publication in ACM-SIAM Symposium on Discrete Algorithms (SODA)' 2015.}}

\author{ 
	Sayan Bhattacharya\thanks{Email: {\tt bsayan@imsc.res.in}. The Institute of Mathematical Sciences, Chennai, India. This work was done while the author was in Faculty of Computer Science, University of Vienna, Austria. The research leading to these results has received funding from the European Research Council under the European Union's Seventh Framework Programe (FP7/2007-2013) / ERC Grant Agreement number \ 340506.}
	\and Monika Henzinger\thanks{Email: {\tt monika.henzinger@univie.ac.at}. Faculty of Computer Science, University of Vienna, Austria. The research leading to these results has received funding from the European Unions Seventh Framework Programme (FP7/2007-2013) under grant agreement \ 317532 and from
the European Research Council under the European Union's Seventh Framework Programme (FP7/2007-2013) / ERC Grant Agreement number  \ 340506.}       
    \and Giuseppe F. Italiano\footnote{Email: {\tt giuseppe.italiano@uniroma2.it}.    Universit\`a di Roma "Tor Vergata", Rome, Italy.  Partially supported by
MIUR, the Italian Ministry
of Education, University and Research, under Project AMANDA
(Algorithmics for MAssive and Networked DAta).}    
}

\begin{document}


\maketitle
\pagenumbering{roman}

\begin{abstract}
	
 We present the first deterministic data structures for maintaining approximate  minimum vertex cover and maximum matching in a fully dynamic graph $G = (V,E)$, with $|V| = n$ and $|E| =m$,  in $o(\sqrt{m}\,)$ time per update. In particular, for minimum vertex cover
we provide deterministic data structures for maintaining 
 a $(2+\eps)$ approximation in $O(\log n/\eps^2)$ amortized time per update. 
 For maximum matching, we  show how to 
maintain
a $(3+\eps)$ approximation in $O(\min(\sqrt{n}/\epsilon, m^{1/3}/\eps^2))$ {\em amortized} time per update, and
a $(4+\eps)$ approximation in $O(m^{1/3}/\eps^2)$ {\em worst-case} time per update.
Our data structure for fully dynamic minimum vertex cover is essentially near-optimal and settles an open problem by Onak and Rubinfeld~\cite{OnakR10}.
\end{abstract}

\newpage
\setcounter{tocdepth}{3}
\tableofcontents

\newpage
\pagenumbering{arabic}

\newcommand{\polylog}{\text{poly}\log}

\section{Introduction}
\label{sec:introduction}

Finding maximum matchings and minimum vertex covers in undirected graphs are classical problems in combinatorial optimization.
Let   $G=(V,E)$ be an undirected graph, with $m = |E|$ edges and $n = |V|$ nodes.
A \emph{matching} in $G$ is a set of vertex-disjoint edges, i.e., no two edges share a common vertex.
A \emph{maximum matching}, also known as maximum cardinality matching, is a matching with the largest possible number of edges. A matching is \emph{maximal} if it is not a proper subset of any other matching in $G$. A subset $V'\subseteq V$ is a \emph{vertex cover} if 
each edge of $G$ is incident to at least one vertex in $V'$.
A \emph{minimum vertex cover} is a vertex cover of smallest possible size. 

The Micali-Vazirani algorithm for maximum matching runs in $O(m\sqrt{n}\,)$ time~\cite{HopcroftK71,MicaliV80}. Using this algorithm, a $(1+\eps)$-approximate maximum matching can be constructed  in $O(m/\eps)$ time~\cite{DuanP10}. Finding a minimum vertex cover, on the other hand,  is  NP-hard. Still, these two problems remain closely related as their LP-relaxations  are duals of each other. Furthermore,  a maximal matching, which can be  computed in $O(m)$ time in a greedy fashion, is known to provide a 2-approximation both to maximum matching and to minimum vertex cover (by using the endpoints of the maximal matching). 
Under the unique games conjecture, the minimum vertex cover cannot be efficiently approximated within any constant factor better than 2~\cite{KhotR08}.
Thus, under the unique games conjecture, the 2-approximation in $O(m)$ time  by the greedy method  is  the optimal guarantee for this problem.

In this paper, we consider a dynamic setting, where the input graph is being updated via a sequence of edge insertions/deletions. The goal is to design data structures that are capable of maintaining the
solution to an optimization problem faster than recomputing it from scratch after each
update. If $P \ne NP$ we cannot achieve polynomial time updates for minimum vertex cover. We also observe that achieving fast update times for maximum matching appears to be a particularly difficult task:  
in this case, an update
bound of $O(\polylog (n))$ would be a breakthrough, since it would
immediately improve the longstanding bounds of various static
algorithms~\cite{HopcroftK71,Madry13,MicaliV80,MuchaS04}. The best known update bound for dynamic maximum matching is obtained by a randomized data structure of
Sankowski~\cite{Sankowski07}, which has  
 $O(n^{1.495})$ time per update. In this scenario, if one wishes to achieve fast update times
for dynamic maximum matching or minimum vertex cover,
approximation appears to
be inevitable.
Indeed, in the last  few years there has been a growing interest in designing efficient dynamic data structures 
for maintaining approximate solutions to both these  problems.

\subsection{Previous work.}
A maximal matching can be maintained
in $O(n)$ worst-case update time   by  a trivial deterministic algorithm. Ivkovi\'{c} and Lloyd~\cite{IvkovicL93} showed how to improve this bound to $O((n + m)^{\sqrt{2}/2})$.
Onak and Rubinfeld~\cite{OnakR10} designed a randomized data structure that maintains constant factor approximations  to maximum matching and to minimum vertex cover in $O(\log^2 n)$ amortized time per update with high probability, with the approximation factors being large constants. Baswana, Gupta and Sen~\cite{BaswanaGS11} improved these bounds by showing that a maximal matching, and thus a 2-approximation of maximum matching and minimum vertex cover, can be maintained in a dynamic graph in amortized $O(\log n)$ update time with high
probability.

Subsequently, turning to deterministic data structures, Neiman
and Solomon~\cite{NeimanS13} showed that a $3/2$-approximate maximum matching  can be
maintained dynamically  in $O(\sqrt{m}\,)$ worst-case time per update. Their data structure maintains a maximal matching and thus achieves the same update bound  also for 2-approximate minimum vertex cover.
Furthermore, Gupta and Peng~\cite{GuptaP13} presented a deterministic data structure to 
maintain a $(1 + \eps)$ approximation of a maximum matching  in $O(\sqrt{m}/\eps^{2})$ worst-case time per update. We also note that Onak and Rubinfeld~\cite{OnakR10} gave a deterministic data structure that maintains an $O(\log n)$-approximate minimum vertex cover in $O(\log^2 n)$ amortized update time. See Table~\ref{table:perspective} for a summary of these results.

Very recently, Abboud and Vassilevska Williams~\cite{Abboud14} showed a conditional lower bound on the performance of any dynamic matching algorithm. There exists an integer $k \in [2,10]$ with the following property: if the dynamic algorithm maintains a matching with the property
that every augmenting path in the input graph (w.r.t.~the matching) has length at least $(2k-1)$,  then an amortized update time of $o(m^{1/3})$ for the algorithm will violate the $3$-SUM conjecture (which states that the $3$-SUM problem on $n$ numbers cannot be solved in $o(n^2)$ time).

\subsection{Our results.}
From the above discussion, it is clear that for both fully dynamic constant approximate maximum matching and minimum vertex cover, there is a huge gap between state of the art deterministic and randomized performance guarantees: the former gives $O(\sqrt{m}\,)$ update time, while the latter gives $O(\log n)$ update time. Thus, it seems natural to ask  whether the $O(\sqrt{m}\,)$ bound achieved in~\cite{GuptaP13,NeimanS13} is a natural barrier for deterministic data structures. 
In particular, in their pioneering work on these problems, Onak and Rubinfeld~\cite{OnakR10}  asked: 
\begin{itemize}
\item ``{\em Is there a deterministic data structure that achieves a constant approximation factor with polylogarithmic update time?}''
\end{itemize}

We answer  this question in the affirmative by presenting a deterministic data structure that maintains  a $(2+\eps)$-approximation of a minimum vertex cover in $O(\log n/\eps^2)$ amortized time per update. Since 
it is impossible to get better than 2-approximation for minimum vertex cover  in polynomial time, our data structure is near-optimal (under the unique games conjecture). As a by product of our approach, we can also maintain, deterministically, a $(2+\eps)$-approximate maximum {\em fractional} matching in $O(\log n/\eps^2)$ amortized update time. 
Note that 
the vertices of the fractional matching polytope of a graph are known to be half integral, i.e., they have only $\{0,1/2,1\}$ coordinates (see, e.g., \cite{LovaszP86}). This implies immediately that 
the value of any fractional matching is at most $3/2$ times the value of the maximum  integral matching. Thus, it follows that we can maintain {\em the value} of the maximum (integral) matching within a factor of $(2+\eps) \cdot (3/2) = (3+O(\eps))$, deterministically, in $O(\log n/\eps^2)$ amortized update time.

Next, we focus on the problem of maintaining an integral matching in a dynamic setting. For this problem, we  show how to 
maintain
a $(3+\eps)$-approximate maximum matching in $O(\min(\sqrt{n}/\epsilon, m^{1/3}/\eps^2))$ amortized time per update, and
a $(4+\eps)$-approximate maximum matching in $O(m^{1/3}/\eps^2)$ worst-case time per update.
Since $m^{1/3} = o(n)$, we provide the first deterministic data structures for dynamic matching  whose  update time  is  sublinear in the number of nodes.   Table~\ref{table:perspective} puts our main results in perspective  with previous work.

\begin{table*}
\begin{large}
\begin{center}
\begin{tabular}{||c||c|c|c|c||}
\hline

{\textbf{Problem}}  & {\textbf{Approximation}} & {\textbf{Update}} & {\textbf{Data}}  & {\textbf{Reference}}\\ 
         & {\textbf{Guarantee}} & {\textbf{Time}} & {\textbf{Structure}} &  \\	\hline	 
          	
\hline\hline
MM \& MVC & $O(1)$ & $O(\log^2 n)$ amortized & randomized &  \cite{OnakR10} \\
\hline
MM \& MVC  & $2$ & $O(\log n)$ amortized & randomized &  \cite{BaswanaGS11} \\
\hline
MM  & $1.5$ & $O(\sqrt{m}\,)$ worst-case & deterministic &  \cite{NeimanS13} \\
\hline
MVC  & $2$ & $O(\sqrt{m}\,)$ worst-case & deterministic &  \cite{NeimanS13} \\
\hline
MM  &$1+\eps$ & $O(\sqrt{m}/\eps^2)$ worst-case & deterministic &  \cite{GuptaP13} \\ 
\hline
\hline
MVC  & $2+\eps$ & $O(\log n/\eps^2)$ amortized & deterministic & {This paper} \\
\hline
MM  & $3+\eps$ & $O(\sqrt{n}/\eps)$ amortized & deterministic & {This paper} \\
\hline
MM  & $3+\eps$ & $O(m^{1/3}/\eps^2)$ amortized & deterministic & {This paper} \\
\hline
MM  & $4+\eps$ & $O(m^{1/3}/\eps^2)$ worst-case & deterministic & {This paper} \\
\hline
\hline
\end{tabular}
\end{center}
\end{large}
\caption{Dynamic data structures for approximate (integral) maximum matching (MM) and minimum vertex cover (MVC).}
\label{table:perspective} 
\end{table*}



\subsection{Our techniques.}
To see why it is difficult to deterministically maintain a dynamic (say maximal) matching, consider the scenario  when a matched edge incident to a node $u$ gets deleted from the graph. To recover from this deletion, we have to scan through the adjacency list of $u$ to check if it has any free neighbor $z$. This takes time proportional to the degree of $u$, which can be $O(n)$. Both the papers~\cite{BaswanaGS11,OnakR10} use randomization to circumvent this problem. Roughly speaking, the idea is to match the node $u$ to one of its free neighbors $z$ picked {\em at random}, and show that even if this step takes $O(\text{deg}(u))$ time, in expectation the newly matched edge $(u,z)$ survives  the next $\text{deg}(u)/2$ edge deletions in the graph  (assuming that the adversary is not aware of the random choices made by the data structure). This is used to bound the amortized update time.

Our key insight is that we can maintain a large {\em fractional matching} deterministically. Suppose that in this fractional matching, we pick each edge incident to $u$ to an extent of $1/\text{deg}(u)$. These edges together contribute at most one to the objective. Thus, we do not have to do anything for the next $\text{deg}(u)/2$  edge deletions  incident to $u$, as these deletions reduce the contribution of $u$ towards the objective by at most a factor of two. This gives us the desired amortized bound. 
Inspired by this observation, we  take a closer look at the framework of Onak and Rubinfeld~\cite{OnakR10}. Roughly speaking, they maintain a hierarchical partition of the set of nodes $V$ into $O(\log n)$ levels such  that the nodes in all but the lowest level, taken together, form a valid vertex cover $V^*$. In addition, they maintain a matching $M^*$  as a {\em dual certificate}. Specifically, they show that  $|V^*| \leq \lambda \cdot |M^*|$ for some  constant $\lambda$, which implies that $V^*$ is a $\lambda$-approximate minimum vertex cover. Their data structure is randomized since, as discussed above, it is particularly difficult to maintain the matching $M^*$ deterministically in a dynamic setting. To make the data structure deterministic, instead of $M^*$, we maintain a {\em fractional matching} as a dual certificate. Along the way, we improve the amortized update time of~\cite{OnakR10} from $O(\log^2 n)$ to $O(\log n/\eps^2)$, and their approximation guarantee from some large constant $\lambda$ to $2+\eps$. 

Our approach  gives near-optimal bounds for fully dynamic minimum vertex cover and, as we have already remarked,  it  maintains a {\em fractional matching}. Next, we consider the problem of maintaining an approximate maximum {\em integral matching}, for which we are able to provide deterministic data structures with improved (polynomial) update time. Towards this end, we introduce the concept of a \emph{kernel} of a graph, which we believe is of independent interest. Intuitively, a kernel is a subgraph with two important properties: (i) each node has bounded degree in the kernel, and (ii) a kernel approximately preserves the size of the maximum matching in the original graph. Our key contribution is to show that a kernel always exists, and that it can be maintained efficiently in a dynamic graph undergoing a  sequence of edge updates.

\section{Deterministic Fully Dynamic Vertex Cover}
\label{appendix:vertex-cover}

The input graph $G = (V,E)$ has $|V| = n$ nodes and zero edges in the beginning. Subsequently,   it keeps getting updated   due to the insertions of new edges and the deletions of already existing edges. The edge updates, however, occur one at a time, while the set  $V$ remains fixed. The goal is to maintain an approximate vertex cover of $G$ in this fully dynamic setting.

In Section~\ref{sec:vc:partition}, we introduce the notion of an $(\alpha,\beta)$-partition of  $G = (V,E)$. This is a hierarchical  partition of the set  $V$ into $L+1$ levels, where $L = \lceil\log_{\beta} (n/\alpha)\rceil$ and $\alpha, \beta > 1$ are two parameters (Definition~\ref{def:vc:partition}). If the $(\alpha,\beta)$-partition satisfies an additional property (Invariant~\ref{inv:vc:1}), then from it we can easily derive a $2\alpha\beta$-approximation to the minimum vertex cover (Theorem~\ref{th:vc:structure}).

 In Section~\ref{sec:vc:algo:main}, we outline a natural deterministic algorithm for maintaining such an $(\alpha,\beta)$-partition. In Section~\ref{sec:vc:datastructures}, we present the relevant data structures that are required for implementing the algorithm. We will analyze its amortized update time using an extremely fine tuned potential function:  In Section~\ref{sec:vc:updatetime}, we present a high level overview of this approach.  Section~\ref{sec:vc:algo} contains a detailed implementation of the algorithm using the data structures from Section~\ref{sec:vc:datastructures}. Finally, in Section~\ref{sec:vc:analysis}, we give a complete analysis of the amortized update time of the algorithm. Specifically, we show that for  $\alpha = 1+3\epsilon$ and $\beta = 1+\epsilon$, the algorithm takes $O\left((t/\epsilon) \log_{1+\epsilon} n\right)$ time to handle $t$ edge updates starting from an empty graph (see Theorem~\ref{th:vc:runtime}). This leads to the main result of this section, which is  summarized in the theorem below.

\begin{theorem}
\label{th:vc:main}
For every $\epsilon \in (0,1]$, we can deterministically maintain a $(2+\epsilon)$-approximate vertex cover in a fully dynamic graph, the amortized update time being $O(\log n/\epsilon^2)$.
\end{theorem}

\subsection{The ($\alpha,\beta$)-partition and its properties.}
\label{sec:vc:partition}

\begin{definition}
\label{def:vc:partition}
An {\em $(\alpha,\beta)$-partition} of the graph $G$ partitions its node-set $V$ into  subsets  $V_0 \ldots V_L$, where $L = \lceil\log_{\beta} (n/\alpha)\rceil$ and $\alpha, \beta > 1$. For $i \in \{0, \ldots, L\}$, we identify the subset $V_i$ as the $i^{th}$ {\em level} of this partition, and denote the {\em level of a node} $v$ by $\ell(v)$. Thus, we have $v \in V_{\ell(v)}$ for all $v \in V$.  Furthermore, the partition assigns a weight $w(u,v) = \beta^{-\max(\ell(u),\ell(v))}$ to every edge $(u,v) \in V$.
\end{definition}
  
Define $\N_v$  to be the set of neighbors  of a node  $v \in V$. Given an $(\alpha,\beta)$-partition,
 let $\N_v(i) \subseteq \N_v$ denote the set of neighbors of $v$ that are in the $i^{th}$ level, and
 let $\N_v(i,j) \subseteq \N_v$ denote the set of neighbors of $v$ whose levels are in the range $[i,j]$.
\begin{equation}
\N_v =   \{ u \in V : (u,v) \in E\} \ \ \forall v \in V. \label{eq:symbol:1} 
\end{equation}
\begin{equation}
\N_v(i)  =   \{ u \in \N_v \cap V_i\} \ \ \forall v \in V; i \in \{0,\ldots,L\} \label{eq:symbol:2} 
\end{equation}
\begin{equation}
\N_v(i,j)  =   \bigcup_{k = i}^j \N_v(k) \ \ \forall v \in V; i,j \in \{0,\ldots,L\}, i \leq j. \label{eq:symbol:3} 
\end{equation}
 
Similarly,  define the notations $\dd_v$, $\dd_v(i)$ and $\dd_v(i,j)$. Note that $\dd_v$ is the degree of a node $v \in V$.
\begin{equation}
\dd_v   =   |\N_v|  \label{eq:symbol:4} 
\end{equation}
\begin{equation}
 \dd_v(i)   =   |\N_v(i)|  \label{eq:symbol:5} 
 \end{equation}
\begin{equation}
 \dd_v(i,j)  =  |\N_v(i,j)|  \label{eq:symbol:6}
 \end{equation}

Given an $(\alpha,\beta)$-partition, let $W_v = \sum_{u \in \N_v} w(u,v)$ denote the total weight a node $v \in V$ receives from the edges incident to it.  We also define the notation $W_v(i)$. It gives the total weight the node $v$ would receive from the edges incident to it, {\em if the node  $v$ itself were to go to the $i^{th}$ level}. Thus, we have $W_v = W_v(\ell(v))$. Since the weight of an edge $(u,v)$ in the hierarchical partition is given by $w(u,v) = \beta^{-\max(\ell(u),\ell(v))}$, we derive the following equations for all nodes $v \in V$.
\begin{equation}
W_v =  \sum_{u \in \N_v} \beta^{-\max(\ell(u),\ell(v))}. \label{eq:symbol:7} 
\end{equation}
\begin{equation}
W_v(i)  =   \sum_{u \in \N_v}  \beta^{-\max(\ell(u),i)} \ \ \forall i \in \{0,\ldots,L\}. \label{eq:symbol:8} 
\end{equation}

\begin{lemma}
\label{lm:partition}
Every $(\alpha,\beta)$-partition of the graph $G$ satisfies the following conditions for all nodes $v \in V$.
\begin{equation}
W_v(L) \leq \alpha \label{eq:lm:partition:1} 
\end{equation}
\begin{equation}
W_v(L)  \leq  \cdots  \leq W_v(i)  \leq \cdots \leq W_v(0)   \label{eq:lm:partition:2} 
\end{equation}
\begin{equation}
W_v(i) \leq \beta \cdot W_v(i+1)  \ \ \forall   i \in \{0,\ldots,L-1\}. \label{eq:lm:partition:3}
\end{equation}
\end{lemma}

\begin{proof}
Fix any $(\alpha,\beta)$-partition and any node $v \in V$.  We prove the first part of the lemma as follows.
\begin{eqnarray*}
W_v(L) = \sum_{u \in \N_v}  \beta^{-\max(\ell(u),L)} \\
 = \sum_{u \in \N_v}  \beta^{-L} \leq n \cdot \beta^{-L} \leq n \cdot \beta^{-\log_{\beta}(n/\alpha)} = \alpha.
\end{eqnarray*}

We now fix any level $i \in \{0,\ldots, L-1\}$ and show that the $(\alpha,\beta)$-partition satisfies equation~\ref{eq:lm:partition:2}.
\begin{eqnarray*}
W_v(i+1) = \sum_{u \in \N_v}  \beta^{-\max(\ell(u),i+1)} \\
\leq \sum_{u \in \N_v}  \beta^{-\max(\ell(u),i)} = W_v(i).
\end{eqnarray*}

Finally, we prove equation~\ref{eq:lm:partition:3}.
\begin{eqnarray*}
W_v(i) = \sum_{u \in \N_v}  \beta^{-\max(\ell(u),i)} = \beta \cdot  \sum_{u \in \N_v}  \beta^{-1-\max(\ell(u),i)} \\
\leq  \beta \cdot  \sum_{u \in \N_v}  \beta^{-\max(\ell(u),i+1)} = \beta \cdot W_v(i+1)
\end{eqnarray*}
\end{proof}

Fix any node $v \in V$, and focus on the value of $W_v(i)$ as we go down from the highest level $i = L$ to the lowest level $i = 0$.  Lemma~\ref{lm:partition} states that  $W_v(i) \leq \alpha$ when $i = L$,  that $W_v(i)$ keeps increasing  as we go down the levels one after another,  and that $W_v(i)$ increases by at most a  factor of $\beta$ between consecutive levels. 

\medskip

We will maintain a specific type of $(\alpha,\beta)$-partition,  where each node is assigned to a level in  a way that satisfies Invariant~\ref{inv:vc:1}.

\begin{invariant}
\label{inv:vc:1}
For every node $v \in V$, if $\ell(v) = 0$, then $W_v \leq \alpha \cdot \beta$. Else if $\ell(v) \geq 1$, then $W_v \in [1, \alpha \beta]$.
\end{invariant}

Consider any $(\alpha,\beta)$-partition satisfying Invariant~\ref{inv:vc:1}. Let $v \in V$ be a node in this partition that is at level $\ell(v) = k \in \{0, \ldots, L\}$. It follows that $\sum_{u \in \N_v(0,k)} w(u,v) = |\N_v(0,k)| \cdot \beta^{-k} \leq W_v \leq \alpha \beta$. Thus, we infer that $|\N_v(0,k)| \leq \alpha \beta^{k+1}$. In other words, Invariant~\ref{inv:vc:1} gives an upper bound on the number of neighbors a node $v$ can have that lie on or below $\ell(v)$. We will crucially use this property in the analysis of our algorithm.


\begin{theorem}
\label{th:vc:structure}
Consider an $(\alpha,\beta)$-partition of the graph $G$ that satisfies Invariant~\ref{inv:vc:1}. Let $V^* = \{ v \in V : W_v \geq 1\}$  be the set of nodes with weight at least one. The set $V^*$ is a feasible vertex cover in $G$. Further,  the size of the set $V^*$ is at most $2\alpha\beta$-times the size of the minimum-cardinality vertex cover in $G$. 
\end{theorem}

\begin{proof}
Consider any edge $(u,v) \in E$. We claim that at least one of its endpoints belong to the set $V^*$. Suppose that the claim is false and we have $W_u < 1$ and $W_v < 1$.  If this is the case, then Invariant~\ref{inv:vc:1} implies that $\ell(u) = \ell(v) = 0$ and $w(u,v) = \beta^{-\max(\ell(u),\ell(v))} = 1$. Since $W_u \geq w(u,v)$ and $W_v \geq w(u,v)$, we get $W_u \geq 1$ and $W_v \geq 1$, and this leads  to a contradiction. Thus, we infer that the set  $V^*$ is a feasible vertex cover in the graph $G$.

Next, we construct a {\em fractional matching} $M_f$ by picking every edge $(u,v) \in E$ to an extent of $x(u,v) = w(u,v)/(\alpha\beta) \in [0,1]$. Since for all nodes $v \in V$, we have $\sum_{u \in \N_v} x(u,v)  = \sum_{u \in \N_v} w(u,v)/(\alpha\beta) = W_v/(\alpha\beta) \leq 1$, we infer that  $M_f$ is a valid fractional matching in $G$. The size of this matching is given by $|M_f| = \sum_{(u,v) \in E} x(u,v) = (1/(\alpha \beta)) \cdot \sum_{(u,v) \in E} w(u,v)$. We now bound the size of  $V^*$ in terms of $|M_f|$.
\begin{eqnarray*}
|V^*|  =  \sum_{v \in V^*} 1   \leq  \sum_{v \in V^*} W_v  =  \sum_{v \in V^*} \sum_{u \in \N_v} w(u,v) \\
 \leq  \sum_{v \in V} \sum_{u \in \N_v} w(u,v)  =   2 \cdot \sum_{(u,v) \in E} w(u,v) =  (2\alpha \beta) \cdot |M_f|
\end{eqnarray*}
The approximation guarantee now follows from the LP duality between minimum fractional vertex cover and maximum fractional matching.
\end{proof}

\paragraph{Query time.}
We store the nodes $v$ with $W_v \geq 1$ as a separate list. Thus, we can report the set of nodes in the vertex cover in $O(1)$ time per node. Using appropriate pointers, we can report in $O(1)$ time whether or not a given node is part of this vertex cover. In $O(1)$ time we can also report the size of the vertex cover.

\subsection{Handling the insertion/deletion of an edge.}
\label{sec:vc:algo:main}

A node is called {\em dirty} if it violates Invariant~\ref{inv:vc:1}, and {\em clean} otherwise. Since the graph $G = (V,E)$ is initially empty,  every node is clean and at level zero before the first update in $G$. Now consider the time instant just prior to the $t^{th}$ update in $G$. By induction hypothesis, at this instant every node is clean. Then  the $t^{th}$ update takes place, which  inserts (resp. deletes)  an edge $(x,y)$ in $G$ with weight $w(x,y) = \beta^{-\max(\ell(x),\ell(y))}$. This increases (resp. decreases) the weights $W_x, W_y$ by $w(x,y)$. Due to this change, the nodes $x$ and $y$ might become dirty.  To recover from this, we call the subroutine in Figure~\ref{fig:vc:dirty:main}.

\begin{figure}[htbp]
\centerline{\framebox{
\begin{minipage}{5.5in}
\begin{tabbing}
01.   \=  {\sc While} there exists a dirty node  $v$ \\
02.  \>  \ \ \ \  \= {\sc If} $W_v > \alpha \beta$, {\sc Then} \\
\> \> \qquad  // {\em If true, then by equation~\ref{eq:lm:partition:1}  $\ell(v) < L$.} \\
03.  \> \> \ \ \ \ \ \ \ \ \ \= Increment the level of $v$  \\
\> \> \> by setting $\ell(v) \leftarrow \ell(v)+1$. \\
04.  \> \> {\sc Else if} ($W_v < 1$ and $\ell(v) > 0$), {\sc Then} \\
05.  \>  \> \> Decrement the level of $v$ \\
\> \> \> by setting $\ell(v) \leftarrow \ell(v)-1$. 
\end{tabbing}
\end{minipage}
}}
\caption{\label{fig:vc:dirty:main} RECOVER().}
\end{figure}

Consider any node $v \in V$ and suppose that $W_v = W_v(\ell(v)) > \alpha \beta$. In this event, equation~\ref{eq:lm:partition:1} implies that $W_v(L) < W_v(\ell(v))$ and hence we have $L > \ell(v)$. In other words, when the procedure described in Figure~\ref{fig:vc:dirty:main} decides to increment the level of a dirty node $v$ (Step 03), we know for sure that the current level of $v$ is strictly less than $L$ (the highest level in the $(\alpha,\beta)$-partition).

Next, consider a node $z \in \N_v$. If we change $\ell(v)$, then this may change the weight $w(v,z)$,  and this in turn may change the weight $W_z$. Thus, a single iteration of the {\sc While} loop in Figure~\ref{fig:vc:dirty:main} may lead to some clean nodes becoming dirty, and some other dirty nodes becoming clean.   If and when  the {\sc While} loop terminates, however, we are guaranteed that every node is clean and that Invariant~\ref{inv:vc:1} holds.

\paragraph{Comparison with the framework of Onak and Rubinfeld~\cite{OnakR10}.} As  described below, there are two significant differences between our framework and that of~\cite{OnakR10}. Consequently, many of the technical details of our approach (illustrated in Section~\ref{sec:vc:analysis})
differ from the proof in~\cite{OnakR10}.

First, in the hierarchical partition of~\cite{OnakR10}, the invariant for a node $y$  consists of $O(L)$ constraints: for each level $i \in \{\ell(y), \ldots, L\}$, the quantity $|\N_y(0,i)|$ has to lie within a certain range. This is the main reason for their  amortized update time being $\Theta(\log^2 n)$. Indeed when a node $y$ becomes dirty, unlike in our setting, they have to spend $\Theta(\log n)$ time just to figure out the new level of $y$. 

Second, along with the hierarchical partition, the authors in~\cite{OnakR10} maintain a matching as a {\em dual certificate}, and show that the size of this matching is within a constant factor of the size of their vertex cover. As pointed out in Section~\ref{sec:introduction}, this is the part where they crucially need to use randomization, as till date there is no deterministic data structure for maintaining a large matching  in polylog amortized update time. We bypass this barrier  by implicitly maintaining a {\em fractional matching} as a dual certificate. Indeed, the weight $w(y,z)$ of an edge $(y,z)$ in our hierarchical partition, after suitable scaling, equals the fractional extent by which the edge $(y,z)$ is included in our fractional matching. 

\subsection{Data structures.}
\label{sec:vc:datastructures}

We now describe the  data structures that we use to implement   the algorithm outlined in Section~\ref{sec:vc:algo:main}

\begin{itemize}
\item We maintain the following data structures for each node $v \in V$.
\begin{itemize}
\item A counter $\text{{\sc Level}}[v]$ to keep track of the current level of $v$. Thus, we set $\text{{\sc Level}}[v] \leftarrow \ell(v)$.
\item A counter $\text{{\sc Weight}}[v]$ to keep track of the weight of  $v$. Thus, we set $\text{{\sc Weight}}[v] \leftarrow W_v$.
\item For every level $i > \text{{\sc Level}}[v]$,  the set of nodes $\N_v(i)$ as a doubly linked list $\text{{\sc Neighbors}}_v[i]$.  For every level $i \leq \text{{\sc Level}}[v]$, the  list $\text{{\sc Neighbors}}_v[i]$ is empty.
\item For level $i = \text{{\sc Level}}[v]$, the set of nodes $\N_v(0,i)$   as a doubly linked list $\text{{\sc Neighbors}}_v[0,i]$.  For every level $i \neq \text{{\sc Level}}[v]$, the list $\text{{\sc Neighbors}}_v[0,i]$ is empty.
\end{itemize}
\item When the  graph $G$ gets updated due to an edge insertion/deletion, we may discover that a node violates Invariant~\ref{inv:vc:1}. Such a node is called {\em dirty}, and we store  the set of such nodes as a doubly linked list $\text{{\sc Dirty-nodes}}$. For every node $v \in V$, we maintain a bit $\text{{\sc Status}}[v] \in \{\text{dirty}, \text{clean}\}$ that indicates if the node is dirty or not. Every dirty node stores a pointer to its position in the list $\text{{\sc Dirty-nodes}}$. 
\item The phrase {\em ``neighborhood lists of $v$''}  refers to the set  $\bigcup_{i=0}^L \left\{ \text{\sc Neighbors}_v[0,i], \text{{\sc Neighbors}}_v[i]\right\}$.
 For every edge $(u,v)$, we maintain two bidirectional pointers:  one links the edge to the position of $v$ in the neighborhood lists of $u$, while the other  links the edge  to the position of $u$ in the neighborhood lists of $v$. 
\end{itemize}

\subsection{Bounding the amortized update time: An overview} 
\label{sec:vc:updatetime}
In Section~\ref{sec:vc:algo}, we present a detailed implementation of our algorithm using the data structures described in Section~\ref{sec:vc:datastructures}.  In Section~\ref{sec:vc:analysis}, we  prove that for any $\epsilon \in [0,1]$,  $\alpha = 1+3\eps$ and $\beta = 1+\eps$,  it takes $O(t \log n/\eps^2)$ time to handle $t$ edge insertions/deletions in $G$ starting from an empty graph. This gives an amortized update time of $O(\log n/\eps^2)$, and by Theorem~\ref{th:vc:structure}, a $(2+14\eps)$-approximation to the minimum vertex cover in $G$. The proof works as follows. First, we note that after an edge insertion or deletion
 the data structure can be updated in time $O(1)$ plus the time to adjust the levels of the nodes, i.e., the time for procedure
RECOVER (see Figure~\ref{fig:vc:dirty:main}). To bound the latter we show that it takes time $\Theta( 1 + \dd_v(0,i))$, when node $v$ changes from level $i$ to level $i+1$ or level $i-1$, and prove the bound on the total time spent in procedure RECOVER using a potential function based argument.

As the formal analysis is quite involved, some high level intuitions are in order. Accordingly, in this section, we describe the main idea behind  a (slightly) simplified variant of the above  argument,  which gives an amortized  bound on  {\em the number of times we have to change the weight of an already existing edge}. 
This number is $\dd_v(0,i)$, when node $v$ changes from level $i$ to level $i+1$ 
and $\dd_v(0,i-1)$, when node $v$ changes from level $i$ to  level $i-1$.\footnote{The proof actually shows a stronger result assuming that
the level change of node $v$  from $i$ to $i-1$ causes  $\Theta(\dd_v(0,i))$ many edges to change their level.}

To simplify the exposition even further, we assume that $\alpha, \beta$ are sufficiently large constants, and describe the main idea behind the proof of Theorem~\ref{th:runtime}. This  implies that  on average we  change the weights of $O(L/\eps) = O(\log n/\eps^2)$ edges per update in $G$, for some sufficiently large constants $\alpha, \beta$.

\begin{theorem}
\label{th:runtime}
Fix two sufficiently large constants $\alpha, \beta$. In the beginning, when  $G$ is an empty graph,   initialize  a counter $\text{\sc Count} \leftarrow 0$. Subsequently, each time we change the weight of an already existing edge in the hierarchical partition, set $\text{{\sc Count}} \leftarrow \text{{\sc Count}} + 1$.  Then $\text{{\sc Count}} = O(t L/\eps)$ just  after we  handle  the $t^{th}$ update in $G$.
\end{theorem}

Define the level of an  edge $(y,z)$ to be $\ell(y,z) = \max(\ell(y),\ell(z))$, and note that the weight  $w(y,z)$ decreases (resp. increases) iff the edge's level $\ell(y,z)$ goes up (resp. down). 
There is a potential associated with both nodes and edges. Note that we use the terms ``tokens'' and ``potential'' interchangeably.

Each edge $e$ has exactly $2 (L - \ell(e))$ tokens. These tokens are assigned as follows.
Whenever a new edge is inserted, it receives  $2 (L - \ell(e))$ tokens. When $e$ moves up a level, it gives one token to each endpoint. Whenever 
$e$  is deleted, it gives one token to each endpoint. Whenever $e$ moves down a level because one endpoint, say $v$, moves down a level, $e$ receives two tokens from $v$.

Initially and whenever a node moves a level higher, it has no tokens. 
{\em Whenever a node $v$ moves up a level}, only its adjacent edges to the same or lower levels have to be updated as their level changes. Recall that each such edge gives 1 token to $v$, which in turn uses this token to pay for updating the edge.
{\em Whenever a node $v$ moves down a level}, say from $k$ to $k-1$, it has at most $\beta^k$ adjacent edges at level $k$ or below. These are all the edges whose level needs to be updated (which costs a token) and whose potential needs to be increased by two tokens. In this case, we show that $v$ has  a number $W_v$ of tokens and $W_v$ is large enough (i)  to pay for the work involved in the update, (ii) to give two tokens to each of the at most $\beta^k$ adjacent edges, {\em and} (iii)
to still have a sufficient number of tokens for being on level $k-1$.

{\em Whenever the level of a node $v$ is not modified but its weight $W_v$  decreases}  because the weight of the adjacent edge $(u,v)$ decreases, the level of $(u,v)$ must have increased and $(u,v)$ gives one token to $v$ (the other one goes to $u$). Note that this implies that a change in $W_v$ by at most $\beta^{-\ell(v)}$ increases the potential of $v$ by 1, i.e., the ``conversion rate'' between weight changes and token changes is
$\beta^{\ell(v)}$.
 {\em Whenever the level of $v$ does not change but its weight $W_v$  increases} as the level of $(u,v)$ has decreased, no tokens are transferred between 
$v$ and $(u,v)$. (Technically the potential of $v$ might fall slightly but the change might be so small that we ignore it.) Formally we achieve
these potential function changes by setting the potential of every node in $V_0$ to 0 and for every
other node to $\beta^{\ell(v)} \cdot \max(0, \alpha - W_v)$.

Thus, the crucial claim is that a node $v$ that moves down to  level $k-1$ has accumulated a sufficient number $W_v$ of tokens,
i.e.,  $W_v \geq  3\beta^k + 
\beta^{k-1} \max(0, \alpha - W_v(k-1)) = X$ (say). We prove this claim by considering two possible cases.

\medskip
\noindent {\em Case 1:} Assume first that $v$ was at level $k-1$ immediately
before being at level $k$.
Recall that, by the definition of the potential function,
 $v$ had {\em no} tokens when it moved up from level $k-1$. However, in this case we know that $W_v$ was least $\alpha \beta$ when $v$ moved up and, thus, after adjusting the weights of its adjacent edges to the level change, $W_v$ was still at least $\alpha$ after the level change to level $k$. Node $v$ only drops to level $k-1$ if $W_v<1$, i.e.,
while being on level $k$ its weight must have dropped by at least $\alpha - 1$. By the above ``conversion rate''  between weight and tokens this means that $v$ must have received at least $ \beta^k (\alpha-1)$ tokens from its adjacent edges while it was on level $k$, which is at least $X$  for large enough $\alpha$. 

\medskip
\noindent {\em Case 2:} Assume next that  node $v$ was at level $k+1$ immediately before level $k$. Right after dropping from  level
$k+1$  node $v$ owned  $\beta^{k} (\alpha - W_v(k))$  tokens. As $v$ has not changed levels since, it did not have to give any tokens to edges
and did not have to pay for any updates of its adjacent edges. Instead it might have received some tokens from inserted or deleted adjacent edges.
Thus, it still owns at least $\beta^{k} (\alpha - W_v(k))$  tokens. As $W_v(k) \le W_v(k-1)$ and $W_v(k) < 1$ when $v$ drops to
level $k-1$, this number of tokens is at least $X$ for $\beta \ge 2$ and $\alpha \ge 3 \beta + 1$.

\medskip
To summarize, whenever an edge is inserted it receives a sufficient number of tokens to pay the cost of future upwards level changes, but also to give a token to its endpoints every time its level increases. These tokens accumulated at the  endpoints are sufficient to pay for level decreases of these endpoints because (a) nodes move up to a level when their weight on the new level is at least $\alpha > 1$ but only  move down when their weight falls below 1 and (b) the weight of edges on the same and lower levels drops by a factor of $\beta$ between two adjacent levels. Thus
$ \beta^{k}(\alpha-1)$ many edge deletions or edge weight decreases of edges adjacent to node $v$
are necessary to cause $v$ to drop from level $k$ to level $k-1$ (each giving one token to $v$), while there are only $\beta^{k-1}$ many edges on levels
below $k$ that need to be updated when $v$ drops. Thus, the cost of $v$'s level drop is $\beta^{k-1}$ and the new potential needed for $v$ on level $k-1$  is
$ \beta^{k-1} (\alpha - 1)$, but $v$ has collected at least $ \beta^{k}(\alpha-1)$ tokens, which, by suitable choice of
$\beta$ and $\alpha$, is sufficient.

\subsection{Implementation details}
\label{sec:vc:algo}

The pseudo-codes for our algorithm, which maintains an $(\alpha,\beta)$-partition that satisfies Invariant~\ref{inv:vc:1}, are given  in Figures~\ref{fig:vc:insert}-\ref{fig:vc:update-down}.

\begin{figure}[htbp]
\centerline{\framebox{
\begin{minipage}{5.5in}
\begin{tabbing}
01. \ \ \ \ \ \ \ \=  $i \leftarrow \text{{\sc Level}}[u']$, \ \   $j \leftarrow \text{{\sc Level}}[v']$. \\
02. \> Call the subroutine UPDATE-LISTS-INSERT($u',v'$). See Figure~\ref{fig:vc:update-insert}. \\
03. \>   $\text{{\sc Weight}}[u'] \leftarrow \text{{\sc Weight}}[u'] + \beta^{-\max(i,j)}$. \\
04. \>   $\text{{\sc Weight}}[v'] \leftarrow \text{{\sc Weight}}[v'] + \beta^{-\max(i,j)}$. \\
05. \> Call the subroutine UPDATE-STATUS($u'$). See Figure~\ref{fig:vc:update-status}.\\
06. \>  Call the subroutine UPDATE-STATUS($v'$). See Figure~\ref{fig:vc:update-status}. \\ 
07. \> {\sc While} the list $\text{{\sc Dirty-nodes}}$ is nonempty: \\
08. \> \ \ \ \ \ \ \ \ \= Let $v$ be the first node in the list $\text{{\sc Dirty-nodes}}$. \\
09. \> \> Call the subroutine FIX($v$). See Figure~\ref{fig:vc:dirty}.
\end{tabbing}
\end{minipage}
}}
\caption{\label{fig:vc:insert} INSERT-EDGE($u',v'$). It updates the data structures upon insertion of  the edge $(u,v)$.}
\end{figure}

\begin{figure}[htbp]
\centerline{\framebox{
\begin{minipage}{5.5in}
\begin{tabbing}
01. \ \ \ \ \ \ \ \=  $i \leftarrow \text{{\sc Level}}[u']$, \ \    $j \leftarrow \text{{\sc Level}}[v']$. \\
02. \> Call the subroutine UPDATE-LISTS-DELETE($u',v')$. See Figure~\ref{fig:vc:update-delete}. \\
03. \>   $\text{{\sc Weight}}[u'] \leftarrow \text{{\sc Weight}}[u'] - \beta^{-\max(i,j)}$. \\
04. \>   $\text{{\sc Weight}}[v'] \leftarrow \text{{\sc Weight}}[v'] - \beta^{-\max(i,j)}$. \\
05. \>  Call the subroutine UPDATE-STATUS($u'$). See Figure~\ref{fig:vc:update-status}. \\
06. \>  Call the subroutine UPDATE-STATUS($v'$).  See Figure~\ref{fig:vc:update-status}. \\ 
07. \> {\sc While} the list $\text{{\sc Dirty-nodes}}$ is nonempty: \\
08. \> \ \ \ \ \ \ \ \ \= Let $v$ be the first node  in the list $\text{{\sc Dirty-nodes}}$.\\
09. \> \> Call the subroutine FIX($v$).  See Figure~\ref{fig:vc:dirty}.
\end{tabbing}
\end{minipage}
}}
\caption{\label{fig:vc:delete} DELETE-EDGE($u',v'$). It updates the data structures upon deletion of the edge $(u,v)$.}
\end{figure}

\paragraph{Initializing the data structures.} In the beginning, the graph $G$ has zero edges, all of its nodes are at level zero, and no node is dirty. At that moment, we ensure that our data structures reflect these conditions.

\paragraph{Handling the insertion of an edge.} When an edge $(u',v')$ is inserted into the graph, we implement the procedure in Figure~\ref{fig:vc:insert}.  In Step 02,   we call the subroutine UPDATE-LISTS-INSERT($u',v'$) to update  the neighborhood lists of the nodes $u'$ and $v'$.  In Steps 03-04, we update the weights of $u'$ and $v'$. In Step 05, we call the subroutine UPDATE-STATUS($u'$), which checks if the new weight of $u'$ satisfies Invariant~\ref{inv:vc:1}, and accordingly, it updates the bit $\text{{\sc Status}}[u']$ and the occurrence of $u'$ in the list $\text{{\sc Dirty-nodes}}$. In Step 06, we perform exactly the same operations on the node $v'$. 

Every node satisfies Invariant~\ref{inv:vc:1} at the beginning of the procedure. However, at the end of Step 06, one or both of the nodes in $\{u',v'\}$ can become dirty. Thus, we run the {\sc While} loop in Steps 07-09 till every node in $G$ becomes clean again. In one iteration of the {\sc While} loop, we pick any dirty node $v$ and call the subroutine FIX($v$), which either increments or decrements the level of $v$ depending on its weight (Figure~\ref{fig:vc:dirty}). The execution of FIX($v$)  can result in more nodes becoming dirty. So the {\sc While} loop can potentially have a large number of iterations. To simply  the analysis, we charge the runtime of each iteration  to the corresponding call to  FIX(.). We separately bound the total time taken by all the calls to FIX(.) in Section~\ref{subsec:analyze:FIX}.

\begin{lemma}
\label{lm:runtime:insert}
Ignoring the time taken by the calls to  FIX(.), an edge-insertion can be handled in $\Theta(1)$ time.
\end{lemma}

\paragraph{Handling the deletion of an edge.}  When an edge $(u',v')$ is deleted from the graph, we implement the procedure in Figure~\ref{fig:vc:delete}, which is very similar to one described above. It first makes all the weights and the neighborhood lists reflect the deletion of the edge. This might lead to one or both the endpoints $\{u',v'\}$ becoming dirty, in which case the procedure keeps on  calling the subroutine FIX(.) till every node becomes clean again. As before, we charge the runtime of the {\sc While} loop in Steps 07-09 (Figure~\ref{fig:vc:delete}) to the respective calls to FIX(.). We  analyze the total runtime of all the calls to FIX(.) in Section~\ref{subsec:analyze:FIX}.

\begin{lemma}
\label{lm:runtime:delete}
Ignoring the time taken by the calls to  FIX(.), an edge-deletion can be handled in $\Theta(1)$ time.
\end{lemma}

\begin{figure}[htbp]
\centerline{\framebox{
\begin{minipage}{5.5in}
\begin{tabbing}
01. \ \ \  \ \   \= {\sc If} $\text{{\sc Weight}}[v] > \alpha \beta$, {\sc Then} \qquad  // {\em If true, then $\text{{\sc Level}}[v] < L$ by Lemma~\ref{lm:partition}.} \\
02. \> \ \  \ \ \ \ \= Call the subroutine MOVE-UP($v$). See Figure~\ref{fig:vc:move-up}.  \\ \\
03. \> {\sc Else if} $\text{{\sc Weight}}[v] < 1$ and $\text{{\sc Level}}[v] > 0$, {\sc Then} \\
04. \> \> Call the subroutine MOVE-DOWN($v$). See Figure~\ref{fig:vc:move-down}.
\end{tabbing}
\end{minipage}
}}
\caption{\label{fig:vc:dirty} FIX($v$). } 
\end{figure}

\paragraph{The subroutine FIX($v$).} The procedure is described in Figure~\ref{fig:vc:dirty}. It is called only if the node $v$ is dirty. This can happen only if either (a) $W_v > \alpha \beta$, or (b) $W_v < 1$ and $\ell(v) > 0$. In the former case,  $W_v$ needs to be reduced to bring it down to the range $[1,\alpha\beta]$. So we call the subroutine MOVE-UP($v$) which increments the level of $v$ by one unit. In the latter case,  $W_v$ needs to be raised to bring it up to the range $[1,\alpha\beta]$. So we call the subroutine MOVE-DOWN($v$) which decrements the level of $v$ by one unit.

\begin{figure}[htbp]
\centerline{\framebox{
\begin{minipage}{5.5in}
\begin{tabbing}
01. \ \ \ \ \ \ \  \ \= $k \leftarrow \text{{\sc Level}}[v]$. \\
02. \> $\text{{\sc Level}}[v] \leftarrow k+1$. \\
03. \>  {\sc While} the list $\text{{\sc Neighbors}}_v[0,k]$ is nonempty: \\
04. \> \ \ \ \ \ \ \ \= Let $u$ be a node that appears in $\text{{\sc Neighbors}}_v[0,k]$. \\
05. \> \>  Call the subroutine UPDATE-LISTS-UP($u,v,k$). See Figure~\ref{fig:vc:update-up}. \\
06. \> \> $\text{{\sc Weight}}[v] \leftarrow \text{{\sc Weight}}[v]  - \beta^{-k} + \beta^{-(k+1)}$. \\
07. \> \> $\text{{\sc Weight}}[u]  \leftarrow \text{{\sc Weight}}[u] - \beta^{-k} + \beta^{-(k+1)}$. \\
08. \>  \> Call the subroutine UPDATE-STATUS($u$). See Figure~\ref{fig:vc:update-status}. \\  
09. \>   In constant time, add all the nodes $u \in \text{{\sc Neighbors}}_v[k+1]$ \\
 \>   to the list $\text{{\sc  Neighbors}}_v[0,k+1]$ by adjusting the relevant pointers, \\
\>  and convert $\text{{\sc Neighbors}}_v[k+1]$ into an empty list. \\ 
10. \> Call the subroutine UPDATE-STATUS($v$). See Figure~\ref{fig:vc:update-status}.
\end{tabbing}
\end{minipage}
}}
\caption{\label{fig:vc:move-up} MOVE-UP($v$). It increments the level of the node $v$ by one.} 
\end{figure}

\paragraph{The subroutine MOVE-UP($v$).} See Figure~\ref{fig:vc:move-up}. Let $k$ be the level of the node $v$ before the call to MOVE-UP($v$). The procedure moves the node $v$ up to level $k+1$. This transition only affects the weights of those edges incident to $v$ whose other endpoints are in $\N_v(0,k)$. In addition, $\N_v(0,k)$ is precisely the set of nodes that should update the occurrences of $v$ in their neighborhood lists. Further, the node $v$ should move all these nodes from the list $\text{{\sc Neighbors}}_v[0,k]$ to the list $\text{{\sc Neighbors}}_v[0,k+1]$. Finally, we need to check if any of the nodes in $\N_v(0,k)$ changes its status from dirty to clean (or the other way round) due to this transition, and accordingly we need to update the list $\text{{\sc Dirty-nodes}}$.  All these operations are performed by the {\sc While} loop in Steps 03-08 (Figure~\ref{fig:vc:move-up}) that runs for $\Theta(1+D_v(0,k))$ time.\footnote{We write $\Theta(1+D_v(0,k))$ instead of $\Theta(D_v(0,k))$ as $D_v(0,k)$ can be zero.}

The nodes in $\N_v(k+1)$ should   be moved from the list $\text{{\sc Neighbors}}_v[k+1]$ to the list $\text{{\sc Neighbors}}_v[0,k+1]$. Unlike the nodes in $\N_v(0,k)$, however, the nodes in $\N_v(k+1)$ need not themselves update the  occurrences  of $v$ in their neighborhood lists. Due to this reason, we can make the necessary changes with regard to these nodes in constant time, as described in Step~09 (see Figure~\ref{fig:vc:move-up}). Finally, in Step~10 we check if the node $v$ satisfies Invariant~\ref{inv:vc:1} after these transformations. If yes, then we remove the node  from the list {\sc Dirty-nodes}. If no, then the node's level needs to be increased further. It remains dirty and is handled in the next call to FIX($v$).  The runtime of  MOVE-UP($v$) is summarized  in the following lemma.

\begin{lemma}
\label{lm:runtime:move-up}
Let $k$ be the level of the node $v$ just prior to a call to MOVE-UP($v$). Then   MOVE-UP($v$) runs for $\Theta(1+D_v(0,k))$ time.
\end{lemma}

\begin{figure}[htbp]
\centerline{\framebox{
\begin{minipage}{5.5in}
\begin{tabbing}
01. \ \ \ \ \ \ \ \  \ \=  $k \leftarrow \text{{\sc Level}}[v]$. \\
02. \> $\text{{\sc Level}}[v] \leftarrow \text{{\sc Level}}[v] - 1$. \\
03. \>  {\sc While} the list $\text{{\sc Neighbors}}_v[0,k]$ is nonempty: \\
04. \> \  \ \ \ \ \ \ \ \= Let $u$ be a node that appears in $\text{{\sc Neighbors}}_v[0,k]$. \\
05. \> \> Call the subroutine UPDATE-LISTS-DOWN($u,v,k$). See Figure~\ref{fig:vc:update-down}. \\ 
06. \> \> {\sc If} $\text{\sc Level}[u] < k$, {\sc Then} \\
07. \> \>  \ \ \ \ \ \ \ \  \ \= $\text{{\sc Weight}}[u] \leftarrow \text{{\sc Weight}}[u] - \beta^{-k} + \beta^{-(k-1)}$. \\
08. \> \> \> $\text{{\sc Weight}}[v] \leftarrow \text{{\sc Weight}}[v] - \beta^{-k} + \beta^{-(k-1)}$. \\
09. \> \> \> Call the subroutine UPDATE-STATUS($u$). See Figure~\ref{fig:vc:update-status}. \\
10. \> Call the subroutine UPDATE-STATUS($v$). See Figure~\ref{fig:vc:update-status}.
\end{tabbing}
\end{minipage}
}}
\caption{\label{fig:vc:move-down} MOVE-DOWN($v$). It decrements the level of the node $v$ by one.} 
\end{figure}

\paragraph{The subroutine MOVE-DOWN($v$).} See Figure~\ref{fig:vc:move-down}. Let $k$ be the level of  $v$ before the call to MOVE-DOWN($v$). The procedure moves the node $v$ down to level $k-1$. This transition only affects the weights of those edges incident to $v$ whose other endpoints are in $\N_v(0,k-1)$. In addition, $\N_v(0,k-1)$ is precisely the set of nodes that should update the  occurrences  of $v$ in their neighborhood lists and that might have to change their statuses from dirty to clean (or the other way round). Finally,   the node $v$ should move  every node in $\N_v(0,k)$ from the list $\text{{\sc Neighbors}}_v[0,k]$ to the list $\text{{\sc Neighbors}}_v[0,k-1]$. All these operations are performed by the {\sc While} loop in Steps 03-09 (see Figure~\ref{fig:vc:move-up}) that runs for $\Theta(1+D_v(0,k))$ time.

In Step~10 (see Figure~\ref{fig:vc:move-down}) we check if the node $v$ satisfies Invariant~\ref{inv:vc:1} under the new circumstances. If yes, then we remove the node  from the list {\sc Dirty-nodes}. If no, then the node's level needs to be decreased further. It remains dirty and is handled in the next call to FIX($v$). 

\begin{lemma}
\label{lm:runtime:move-down}
Let $k$ be the level of the node $v$ just prior to a call to MOVE-DOWN($v$). Then the call to the subroutine MOVE-DOWN($v$) takes $\Theta(1+D_v(0,k))$ time.
\end{lemma}

\begin{figure}[htbp]
\centerline{\framebox{
\begin{minipage}{5.5in}
\begin{tabbing}
01. \ \ \ \ \    \= {\sc If} $1 \leq \text{{\sc Weight}}[u] \leq \alpha \beta$, {\sc Then} \\
02. \> \ \ \ \ \ \ \ \= {\sc If} $\text{{\sc Status}}[u] = \text{dirty}$, {\sc Then} \\
03. \> \>  \ \ \ \ \ \ \ \= Remove the node $u$ from the list {\sc Dirty-nodes}. \\
04. \> \> \> $\text{{\sc Status}}[u] \leftarrow \text{clean}$. \\
05. \> {\sc Else if} $\text{{\sc Weight}}[u] > \alpha \beta$, {\sc Then} \\
06. \> \> {\sc If} $\text{{\sc Status}}[u] = \text{clean}$, {\sc Then} \\
07. \> \> \> Insert the node into the list $\text{{\sc Dirty-nodes}}$. \\
08. \> \> \> $\text{{\sc Status}}[u] \leftarrow \text{dirty}$. \\
09.  \> {\sc Else if} $\text{{\sc Weight}}[u] < 1$ and $\text{{\sc Level}}[u] > 0$, {\sc Then} \\
10. \> \> {\sc If} $\text{{\sc Status}}[u] = \text{clean}$, {\sc Then} \\
11. \> \> \> Insert the node into the list $\text{{\sc Dirty-nodes}}$. \\
12. \> \> \> $\text{{\sc Status}}[u] \leftarrow \text{dirty}$. \\
13.  \> {\sc Else if} $\text{{\sc Weight}}[u] < 1$ and $\text{{\sc Level}}[u] = 0$, {\sc Then} \\
14. \> \ \ \ \ \ \ \ \= {\sc If} $\text{{\sc Status}}[u] = \text{dirty}$, {\sc Then} \\
15. \> \>  \ \ \ \ \ \ \ \= Remove the node $u$ from the list {\sc Dirty-nodes}. \\
16. \> \> \> $\text{{\sc Status}}[u] \leftarrow \text{clean}$. 
\end{tabbing}
\end{minipage}
}}
\caption{\label{fig:vc:update-status} UPDATE-STATUS($u$). The subroutine verifies if the node $u$ is dirty or not, and accordingly, it updates the bit $\text{{\sc Status}}[u]$ and the occurrence of the node $u$ in the list $\text{{\sc Dirty-nodes}}$.}  
\end{figure}

\paragraph{The subroutine UPDATE-STATUS($v$).} The procedure is described in Figure~\ref{fig:vc:update-status}. At every stage of our algorithm, it ensures that the list {\sc Dirty-nodes} contains exactly those nodes that violate Invariant~\ref{inv:vc:1}.

\begin{figure}[htbp]
\centerline{\framebox{
\begin{minipage}{5.5in}
\begin{tabbing}
01. \ \ \ \ \ \ \ \ \ \= $i \leftarrow \text{{\sc Level}}[u']$, \ \  $j \leftarrow \text{{\sc Level}}[v']$. \\
02. \>   {\sc If} $(i = j)$, {\sc Then} \\
03. \>  \ \ \ \ \ \ \ \ \=  Insert node $u'$ into the list $\text{{\sc Neighbors}}_{v'}[0,i]$. \\
04. \> \>  Insert node $v'$ into the list $\text{{\sc Neighbors}}_{u'}[0,i]$. \\
05. \> {\sc Else if} $(i > j)$, {\sc Then} \\
06. \> \> Insert node $u'$ into the list $\text{{\sc Neighbors}}_{v'}[i]$. \\
07. \> \>  Insert node $v'$ into the list $\text{{\sc Neighbors}}_{u'}[0,i]$. \\
08. \> {\sc Else if} $(i < j)$, {\sc Then} \\
09. \> \> Insert node $u'$ into the list $\text{{\sc Neighbors}}_{v'}[0,j]$. \\
10. \> \>  Insert node $v'$ into the list $\text{{\sc Neighbors}}_{u'}[j]$.
\end{tabbing}
\end{minipage}
}}
\caption{\label{fig:vc:update-insert} UPDATE-LISTS-INSERT($u',v'$). 
The edge $(u',v')$ has been inserted into the graph $G = (V,E)$. The subroutine ensures that the neighborhood lists of $u',v'$ reflect this change.}
\end{figure}

\begin{figure}[htbp]
\centerline{\framebox{
\begin{minipage}{5.5in}
\begin{tabbing}
01. \ \ \  \ \ \ \ \ \= $i \leftarrow \text{{\sc Level}}[u']$, \ \ $j \leftarrow \text{{\sc Level}}[v']$. \\
02. \>   {\sc If} $(i = j)$, {\sc Then} \\
03. \>  \ \ \ \ \ \ \ \ \=  Delete node $u'$ from the list $\text{{\sc Neighbors}}_{v'}[0,i]$. \\
04. \> \>  Delete node $v'$ from the list $\text{{\sc Neighbors}}_{u'}[0,i]$. \\
05. \> {\sc Else if} $(i > j)$, {\sc Then} \\
06. \> \> Delete node $u'$ from the list $\text{{\sc Neighbors}}_{v'}[i]$. \\
07. \> \>  Delete node $v'$ from the list $\text{{\sc Neighbors}}_{u'}[0,i]$. \\
08. \> {\sc Else if} $(i < j)$, {\sc Then} \\
09. \> \> Delete node $u'$ from the list $\text{{\sc Neighbors}}_{v'}[0,j]$. \\
10. \> \>  Delete node $v'$ from the list $\text{{\sc Neighbors}}_{u'}[j]$.
\end{tabbing}
\end{minipage}
}}
\caption{\label{fig:vc:update-delete} UPDATE-LISTS-DELETE($u',v'$). 
The edge $(u',v')$ has been deleted from the graph $G = (V,E)$. The subroutine ensures that the neighborhood lists of $u',v'$ reflect this change.}
\end{figure}

\begin{figure}[htbp]
\centerline{\framebox{
\begin{minipage}{5.5in}
\begin{tabbing}
01. \ \ \ \ \ \ \   \= Delete the node $u$ from the list $\text{{\sc Neighbors}}_v[0,k]$. \\
02. \> Insert the node $u$ into the list $\text{{\sc Neighbors}}_v[0,k+1]$. \\
03. \> {\sc If} $\text{{\sc Level}}[u] = k$, {\sc Then} \\
04. \> \ \ \ \ \ \ \ \ \=  Delete the node $v$ from the list $\text{{\sc Neighbors}}_u[0,k]$. \\
05. \> {\sc Else if} $\text{{\sc Level}}[u] < k$, {\sc Then} \\
06. \> \> Delete the node $v$ from the list $\text{{\sc Neighbors}}_u[k]$. \\
07. \> Insert the node $v$ into the list $\text{{\sc Neighbors}}_u[k+1]$.
\end{tabbing}
\end{minipage}
}}
\caption{\label{fig:vc:update-up} UPDATE-LISTS-UP($u,v,k$). 
The level of the node $v$ has changed from $k$ to $k+1$, and we have $u \in \N_v(0,k)$. The subroutine ensures that the neighborhood lists of $u,v$ reflect this change.}
\end{figure}

\begin{figure}[htbp]
\centerline{\framebox{
\begin{minipage}{5.5in}
\begin{tabbing}
01. \ \ \ \ \ \ \   \= Delete the node $u$ from the list $\text{{\sc Neighbors}}_v[0,k]$. \\
02. \> {\sc If} $\text{{\sc Level}}[u] = k$, {\sc Then} \\
03. \> \ \ \ \ \ \ \ \ \=  Insert the node $u$ into the list $\text{{\sc Neighbors}}_v[k]$. \\
04. \> {\sc Else if} $\text{{\sc Level}}[u] = k-1$, {\sc Then} \\
05. \> \> Insert the node $u$ into the list $\text{{\sc Neighbors}}_v[0,k-1]$. \\
06. \> \> Delete the node $v$ from the list $\text{{\sc Neighbors}}_u[k]$. \\
07. \> \> Insert the node $v$ into the list $\text{{\sc Neighbors}}_u[0,k-1]$. \\
08. \> {\sc Else if} $\text{{\sc Level}}[u] < k-1$, {\sc Then} \\
09. \> \> Insert the node $u$ into the list $\text{{\sc Neighbors}}_v[0,k-1]$. \\
10. \> \> Delete the node $v$ from the list $\text{{\sc Neighbors}}_u[k]$. \\
11. \> \> Insert the node $v$ into the list $\text{{\sc Neighbors}}_u[k-1]$. 
\end{tabbing}
\end{minipage}
}}
\caption{\label{fig:vc:update-down} UPDATE-LISTS-DOWN($u,v,k$). 
The level of the node $v$ has changed from $k$ to $k-1$, and we have $u \in \N_v(0,k)$. The subroutine ensures that the neighborhood lists of $u,v$ reflect this change.}
\end{figure}

\paragraph{Answering a query.} We can handle three types of queries. 
\begin{enumerate}
\item {\em Is the node $v \in V$ part of the vertex cover?} 

If $\text{{\sc Weight}}[v] \geq 1$, then we output {\sc Yes}, else we output {\sc No}. The query time is $\Theta(1)$.

\item {\em What is the size of the current vertex cover?}

Without increasing the asymptotic runtime of our algorithm, we can easily modify it  and  maintain an additional counter that will store the number of nodes $v \in V$ with $\text{{\sc Weight}}[v] \geq 1$. To answer this query, we simply return the value of this counter. The query time is $\Theta(1)$.

\item {\em Return all the nodes that are part of the current vertex cover.}

Without increasing the asymptotic runtime of our algorithm, we can easily modify it and maintain an additional doubly linked list that will store all the nodes $v \in V$ with $\text{{\sc Weight}}[v] \geq 1$. To answer this query, we simply go through this list and return  the nodes that appear in it. The query time is $O(|V^*|)$, where $V^*$ is the vertex cover maintained by our algorithm.
\end{enumerate}

\subsection{Bounding the amortized update time: Detailed analysis} 
\label{sec:vc:analysis}

We devote Section~\ref{sec:vc:analysis} to the proof of the following theorem, which bounds the update time of our algorithm. Throughout this section, we  set $\alpha \leftarrow 1+3\epsilon$ and $\beta \leftarrow 1+\epsilon$, where $\epsilon \in (0,1]$ is a small positive constant.

\begin{theorem}
\label{th:vc:runtime}
The algorithm in Section~\ref{sec:vc:algo} maintains an $(\alpha,\beta)$-partition that satisfies Invariant~\ref{inv:vc:1}. For every $\epsilon \in (0,1], \alpha = 1+3\epsilon$ and $\beta = 1+\epsilon$, the algorithm takes $O\left((t/\epsilon) \log_{1+\epsilon} n\right)$ time to handle $t$ edge updates starting from an empty graph.
\end{theorem}

Consider the following thought experiment. We have a {\em bank account}, and initially, when there are no edges in the graph,   the bank account has zero balance.  For each subsequent edge insertion/deletion,   at most $20 L/\epsilon$ dollars are deposited to the bank account; and for each unit of computation performed by our algorithm, at least one   dollar is withdrawn from it. We show that the bank account never runs out of money, and this gives a running time bound of $O(t L/\epsilon) = O\left((t/\epsilon) \log_{\beta}(n/\alpha)\right) = O\left((t/\epsilon)\log_{1+\epsilon} n\right)$ for handling $t$ edge updates starting from an empty graph.

Let $\B$ denote the total amount of money (or potential) in the bank account at the present moment. We keep track of $\B$  by distributing an $\epsilon$-fraction of it among the nodes and the current set of edges in the graph. 
\begin{equation}
\label{eq:vc:potential:0}
\B = (1/\epsilon) \cdot \left(\sum_{e \in E} \Phi(e) + \sum_{v \in V} \Psi(v)\right)
\end{equation}

In the above equation, the amount of money (or potential) associated with an edge $e \in E$ is given by $\Phi(e)$, and the amount of money (or potential) associated with a node $v \in V$ is given by $\Psi(v)$. To ease notation, for each edge $e = (u,v) \in E$, we use the symbols $\Phi(e), \Phi(u,v)$ and $\Phi(v,u)$ interchangeably. 

We call a node $v \in V$  {\em passive} if we have $\N_v = \emptyset$ throughout the duration of a time interval that starts at the beginning of the algorithm (when $E = \emptyset$) and ends at  the present moment. Let $V_{passive} \subseteq V$  denote the set of all  nodes that are currently passive.

At every point in time, all the potentials $\Phi(u,v), \Psi(v)$    are determined by  the two invariants stated below.

\begin{invariant}
\label{inv:vc:potential:edge}
For every (unordered) pair of node $\{u,v\}$, $u,v \in V$, we have: 
\begin{eqnarray*}
\Phi(u,v)  = \begin{cases}
 (1+\epsilon) \cdot \left(L - \max(\ell(u),\ell(v))\right) & \text{ if } (u,v) \in E; \\
0 & \text{ if } (u,v) \notin E.
\end{cases}
\end{eqnarray*}
\end{invariant}

\begin{invariant}
\label{inv:vc:potential:node}
For every node $v \in V$, we have:
\begin{eqnarray*}
\Psi(v)  =  \begin{cases}
\epsilon \cdot (L- \ell(v))  + \left(\beta^{\ell(v)+1}/(\beta-1)\right) \cdot \max\left(0,\alpha-W_v\right) & \text{ if } v \notin V_{\text{passive}}; \\
0 & \text{ if } v \in V_{\text{passive}}.
\end{cases}
\end{eqnarray*}
\end{invariant}

\paragraph{Initialization.}
When the algorithm starts, the graph has zero edges, all the nodes are at level $0$, and every node is passive. At that moment, Invariant~\ref{inv:vc:potential:node} sets $\Psi(v) = 0$ for all nodes $v \in V$.  Consequently,  equation~\ref{eq:vc:potential:0} implies that the potential $\B$ is also set to zero. This is consistent with our requirement that initially  the bank account ought to have zero balance.

\paragraph{Insertion of an edge.}
By Lemma~\ref{lm:runtime:insert}, the  time taken to handle an edge insertion, ignoring the calls to the FIX(.) subroutine, is $\Theta(1)$. According to our framework,   we are allowed to deposit at most $20L/\epsilon$ dollars  to the bank account, and a withdrawal of one dollar from the same place is sufficient to pay for the computation performed. Thus, due to the Steps 01-06 in Figure~\ref{fig:vc:insert}, the net increase in the potential $\B$ ought to be no more than $20 L/\epsilon - 1$. We  show below that this is indeed the case.

When an edge $(u,v)$ is inserted into the graph,  $\Phi(u,v)$ increases by at most $(1+\epsilon)L$. Next, we consider two possible scenarios to bound the increase in $\Psi(v)$.
\begin{enumerate}
\item The node $v$ was passive prior to the insertion of the edge. Clearly, in this case the node is at level zero, and $\Psi(v) = 0$ before the insertion. After the insertion the node is not passive anymore, and we have $\Psi(v) \leq  \epsilon L + \alpha\beta/(\beta-1) = \epsilon L + (1+3\epsilon)(1+\epsilon)/\epsilon \leq 9L$. Since $\epsilon \in [0,1]$ and $L = \log_{\beta} (n/\alpha) = \log_{1+\epsilon} \left(n/(1+3\epsilon)\right)$, the last inequality holds as long as  $n \geq 20$. To summarize, we conclude that $\Psi(v)$ increases by at most $9L$.
\item The node $v$ was not passive prior to the insertion of the edge. Clearly, in this case the node $v$ is also not passive afterwards. The weight $W_v$ increases and the quantity $\max(0,\alpha - W_v)$ decreases due to the insertion, which means that the potential $\Psi(v)$ decreases.
\end{enumerate}
Similarly, we conclude that  either $\Psi(u)$  increases by at most $9L$ or it actually decreases. The potentials of the remaining nodes and edges do not change. Hence, by equation~\ref{eq:vc:potential:0},  the net increase in   $\B$ is at most $18L/\epsilon$.

\paragraph{Deletion of an edge.}
The analysis  is very similar to the one described above. By Lemma~\ref{lm:runtime:delete}, the  time taken to handle an edge deletion, ignoring the calls to FIX(.), is $\Theta(1)$. According to our framework,  we are allowed to deposit at most $20L/\epsilon$ dollars  to the bank account, and a withdrawal of  one dollar  from the same place is sufficient to pay for the computation performed. Thus, due to the Steps~01-06 in Figure~\ref{fig:vc:delete}, the net increase in the potential  $\B$ ought to be no more than $20 L/\epsilon - 1$. We  show below that this is indeed the case.

When an edge $(u,v)$ is deleted from the graph, the potential $\Phi(u,v)$ decreases. Next, note that  the weight $W_v$ decreases by at most $\beta^{-\ell(v)}$, and so the quantity $\max(0,\alpha - W_v)$ increases by at most $\beta^{-\ell(v)}$. As the node $v$ was not passive before the edge deletion, it is also not passive afterwards.  Thus,  $\Psi(v)$ increases by at most $\left(\beta^{\ell(v)+1}/(\beta-1)\right) \cdot \beta^{-\ell(v)} = 2/\epsilon \leq 2 L$. The last equality holds as long as $n \geq 20$. We similarly conclude that $\Psi(u)$  increases by at most $2L$. The potentials of the remaining nodes and edges do not change. Hence, by equation~\ref{eq:vc:potential:0},  the net increase in   $\B$ is  at most  $4L/\epsilon$.

\medskip

It remains to analyze the total runtime of all the calls to  FIX(.). This is done in Section~\ref{subsec:analyze:FIX}.

\subsubsection{Analysis of the  FIX($v$) subroutine}
\label{subsec:analyze:FIX}

We analyze a single call to the subroutine FIX($v$). Throughout this section, we use the superscript $0$ (resp. $1$) on a symbol to denote its state at the time instant immediately prior to (resp. after) the execution of FIX($v$). Further,  we preface a symbol with  $\delta$ to denote the net decrease in its value due to the call to FIX($v$). For example, consider the potential $\B$.  We have $\B = \B^0$ immediately before FIX($v$) is called, and  $\B = \B^1$ immediately after  FIX($v$) terminates.  We also have $\delta \B = \B^0 - \B^1$.  We will prove the following theorem. 
\begin{theorem}
\label{th:vc:analyze:FIX:main}
We have $\delta \B > 0$ and $\delta \B \geq T$, where $T$ denotes the runtime of FIX($v$). In other words, the money withdrawn from the bank account during the execution of FIX($v$) is sufficient to pay for  the computation performed by FIX($v$). 
\end{theorem}

The call to FIX($v$) affects only the potentials of the nodes  $u \in \N_v \cup \{v\}$ and that of the edges $e \in \{(u,v) : u \in \N_v\}$. This observation, coupled with equation~\ref{eq:vc:potential:0}, gives us the following guarantee.
\begin{equation}
\label{eq:vc:change:1}
\delta \B  = (1/\epsilon) \cdot \left( \delta \Psi(v) +   \sum_{u \in \N_v} \delta \Phi(u,v) + \sum_{u \in \N_v}  \delta \Psi(u) \right)
\end{equation}

The call to  FIX($v$) does not change (a) the neighborhood structure  of the  node $v$, and (b) the level and the overall degree of any node $u \neq v$. 
\begin{eqnarray}
 \N_v^0(i)  & = &  \N_v^1(i) \ \ \text{ for all } i \in \{0, \ldots, L\}. \label{eq:nochange:1} \\
 \ell^0(u) & = & \ell^1(u) \ \ \  \text{ for all } u \in V \setminus \{v\}. \label{eq:nochange:2} \\
D^0_u & = & D^1_u \  \  \text{ for all } u \in V \setminus \{v\}. \label{eq:nochange:3}
\end{eqnarray}

Accordingly, to ease notation we do not  put any superscript on the following symbols, as the quantities they refer to remain the same throughout the duration of FIX($v$).
\begin{eqnarray*}
\begin{cases}
\N_v, D_v. &  \\
\N_v(i), \dd_v(i), W_v(i) &   \text{ for all } i \in \{0, \ldots, L\}. \\
\N_v(i,j), \dd_v(i,j) &  \text{ for all } i,j \in \{0,\ldots, L\}, i \leq j. \\
\ell(u), D_u &  \text{ for all } u \in V \setminus \{v\}.
\end{cases}
\end{eqnarray*}

Since the node $v$ is dirty when the subroutine FIX($v$) is called, it follows that  neither the node $v$ nor  the nodes $u \in \N_v$ are passive.\footnote{A passive node is not adjacent to any edges and, thus, its weight is zero and it has no neighbors.} Applying Invariant~\ref{inv:vc:potential:node}, we get:
\begin{equation}
\label{eq:inv:vc:potential:node}
\Psi(u) = \epsilon \cdot (L- \ell(u))  + \left(\beta^{\ell(u)+1}/(\beta-1)\right) \cdot \max\left(0,\alpha-W_u\right) \ \ \text{ for all nodes } u \in \N_v \cup \{v\}.
\end{equation}

We divide the proof of Theorem~\ref{th:vc:analyze:FIX:main} into two possible cases, depending upon whether the call to FIX($v$) increments or decrements the level of $v$. The main approach to the proof remains the same in each case. We first give an upper bound on the running time $T$. Next, we separately lower bound each of the following quantities:  $\delta \Psi(v)$,  $\delta \Phi(u,v)$ for all $u \in \N_v$, and   $\delta \Psi(u)$ for all $u \in \N_v$. Finally, applying equation~\ref{eq:vc:change:1}, we  derive that $\delta \B \geq T$.

\paragraph{Case 1: The subroutine FIX($v$) increases the level of the node $v$ from $k$ to $(k+1)$.}
\label{sec:FIX:case2}

\begin{lemma}
\label{lm:FIX:case2:1}
We have $T \leq 1+\dd_v(0,k)$.
\end{lemma}

\begin{proof}
 FIX($v$) calls the subroutine MOVE-UP($v$), which, by Lemma~\ref{lm:runtime:move-up}, runs for $\Theta(1+D_v(0,k))$ time.\footnote{Note that $D_v(0,k)$ may be zero. Hence, we bound the runtime of MOVE-UP($v$) by $\Theta(1+D_v(0,k))$ instead of $\Theta(D_v(0,k))$.}
\end{proof}

\begin{lemma}
\label{lm:FIX:case2:2}
We have $\delta \Psi(v) = \epsilon$.
\end{lemma}

\begin{proof}
The subroutine  FIX($v$) calls the subroutine  MOVE-UP($v,j$) in Step~02 (Figure~\ref{fig:vc:dirty}). Hence,   Step 01 (Figure~\ref{fig:vc:dirty}) guarantees that  $W_v^0 = W_v(k) > \alpha \beta$. Next, from Lemma~\ref{lm:partition} we infer that $W_v^1  = W_v(k+1) \geq \beta^{-1} \cdot W_v(k) > \alpha$. Since both $W_v^0, W_v^1 > \alpha$, we get:
\begin{eqnarray*}
\Psi^0(v) & = & \epsilon  \cdot (L- k) + \left(\beta^{k+1}/(\beta-1)\right) \cdot \max(0, \alpha - W_v^0) =   \epsilon \cdot (L- k)  \\
\Psi^1(v) & = & \epsilon  \cdot (L- k-1) + \left(\beta^{k+2}/(\beta-1)\right) \cdot \max(0, \alpha - W_v^1) =   \epsilon \cdot (L- k-1) 
\end{eqnarray*}
It follows that $\delta \Psi(v) = \Psi^0(v) - \Psi^1(v) = \epsilon$.
\end{proof}

\begin{lemma}
\label{lm:FIX:case2:3}
For every node $u \in \N_v$, we have:
\begin{eqnarray*}
\delta \Phi(u,v)  =
\begin{cases}
(1+\epsilon) &  \text{ if } u \in \N_v(0,k); \\
0 &  \text{ if } u \in \N_v(k+1,L). 
\end{cases} 
\end{eqnarray*} 
\end{lemma}

\begin{proof}
If $u \in \N_v(0,k)$, then  we have $\Phi^0(u,v) = (1+\epsilon) \cdot (L - k)$ and $\Phi^1(u,v) =  (1+\epsilon) \cdot (L-k-1)$. It follows that $\delta \Phi(u,v) = \Phi^0(u,v) - \Phi^1(u,v) = (1+\epsilon)$. 

In contrast, if $u \in \N_v(k+1,L)$, then we have  $\Phi^0(u,v) = \Phi^1(u,v) = (1+\epsilon) \cdot (L - \ell(u))$. Hence, we get $\delta \Phi(u,v) = \Phi^0(u,v) - \Phi^1(u,v) = 0$.
\end{proof}

\begin{lemma}
\label{lm:FIX:case2:4}
For every node $u \in \N_v$, we have:
\begin{eqnarray*}
\delta \Psi(u)  \geq
\begin{cases}
-1 &  \text{ if } u \in \N_v(0,k); \\
0 &  \text{ if } u \in \N_v(k+1,L). 
\end{cases} 
\end{eqnarray*} 
\end{lemma}

\begin{proof}
Consider any node $u \in \N_v(k+1,L)$. Since $k < \ell(u)$, we have $w^0(u,v) = w^1(u,v)$, and this implies  that $W_u^0 = W_u^1$. Thus,  we get $\delta \Psi(u) = 0$.

Next, fix any node $u  \in \N_v(0,k)$. Note that   $\delta W_u = \delta w(u,v) = \beta^{-k} - \beta^{-(k+1)} = (\beta-1)/\beta^{k+1}$. Using this observation, and the fact that $\ell(u) \leq k$, we infer that:  
$$\delta \Psi(u) \geq - \left(\beta^{\ell(u)+1}/(\beta-1)\right) \cdot \delta W_u = - \beta^{\ell(u)+1}/\beta^{k+1} \geq -1.$$
\end{proof}

\begin{proof}[Proof of Theorem~\ref{th:vc:analyze:FIX:main} (for Case 1)]
From Lemmas~\ref{lm:FIX:case2:2},~\ref{lm:FIX:case2:3},~\ref{lm:FIX:case2:4} and equation~\ref{eq:vc:change:1}, we derive the following bound.
\begin{eqnarray}
\delta \B & = & (1/\epsilon) \cdot \left(\delta \Psi(v) + \sum_{u \in \N_v} \delta \Phi(u,v) + \sum_{u \in \N_v} \delta \Psi(u)\right)  \nonumber \\
& \geq & (1/\epsilon) \cdot \left(\epsilon + (1+\epsilon) \cdot D_v(0,k) -  D_v(0,k) \right) \nonumber \\
&  = &   1 +   D_v(0,k) \nonumber
\end{eqnarray}

The theorem (for Case 1) now follows from Lemma~\ref{lm:FIX:case2:1}.
\end{proof}

\paragraph{Case 2:  The subroutine FIX($v$) decreases the level of the node $v$ from $k$ to  $k-1$.}
\label{sec:FIX:case3}

\begin{claim}
\label{cl:FIX:case3:degree}
We have $W_v^0 = W_v(k) < 1$ and $D_v(0,k) \leq \beta^k$.
\end{claim}

\begin{proof}
As FIX($v$) calls the subroutine MOVE-DOWN($v$), Step~03 (Figure~\ref{fig:vc:dirty}) ensures that $W_v^0 = W_v(k) < 1$. Since $\ell^0(v) = k$, we have $w^0(u,v) \geq \beta^{-k}$ for all $u \in \N_v$. We conclude that: 
$$1 > W_v^0 \geq \sum_{u \in \N_v(0,k)} w^0(u,v) \geq \beta^{-k} \cdot D_v(0,k).$$ Thus, we get $D_v(0,k) \leq \beta^k$.
\end{proof}

\begin{lemma}
\label{lm:case3:runtime}
We have $T \leq \beta^k$.
\end{lemma}

\begin{proof}
FIX($v$) calls the subroutine  MOVE-DOWN($v$), which runs for $\Theta(1+D_v(0,k))$ time by Lemma~\ref{lm:runtime:move-down}.\footnote{Since $D_v(0,k)$ may be zero, we bound the runtime of MOVE-DOWN($v$) by $\Theta(1+D_v(0,k))$ instead of $\Theta(D_v(0,k))$.} The lemma now follows from Claim~\ref{cl:FIX:case3:degree}.
\end{proof}

\begin{lemma}
\label{lm:FIX:case3:node:u}
For every node $u \in \N_v$, we have $\delta \Psi(u)  \geq 0$. 
\end{lemma}

\begin{proof}
Fix any node $u \in \N_v$. As the level of the node $v$ decreases from $k$ to $k-1$, we infer that $w^0(u,v) \leq w^1(u,v)$, and, accordingly, we get $W_u^0 \leq W_u^1$. Since $\Psi(u) = \epsilon \cdot (L-\ell(u)) + \beta^{\ell(u)} \cdot \max\left(0,\alpha - W_u\right)$, we derive that $\Psi^0(u) \geq \Psi^1(u)$. Thus, we have $\delta \Psi(u) = \Psi^0(u) - \Psi^1(u) \geq 0$.
\end{proof}

\begin{lemma}
\label{lm:FIX:case3:edge}
For every node $u \in \N_v$, we have:
\begin{eqnarray*}
\delta \Phi(u,v)  = 
\begin{cases}
0 &  \text{ if } u \in \N_v(k,L); \\
-(1+\epsilon)  &  \text{ if } u \in \N_v(0,k-1); 
\end{cases} 
\end{eqnarray*} 
\end{lemma}

\begin{proof}
Fix any node $u \in \N_v$. We prove the lemma by considering two possible scenarios.
\begin{enumerate}
\item We have $u \in \N_v(k,L)$. As FIX($v$) decreases the level of the node $v$ from $k$ to $k-1$,  we infer that  $\Phi^0(u,v) = \Phi^1(u,v) = (1+\epsilon) \cdot (L - \ell(u))$. Hence, we get $\delta \Phi(u,v) = \Phi^1(u,v) - \Phi^0(u,v) = 0$.
\item We have $u \in \N_v(0,k-1)$. Since FIX($v$) decreases the level of node $v$ from $k$ to $k-1$,  we infer that $\Phi^0(u,v) =  (1+\epsilon) \cdot (L - k)$ and $\Phi^1(u,v) = (1+\epsilon) \cdot (L-k+1)$. Hence, we get $\delta \Phi(u,v) = \Phi^1(u,v) - \Phi^0(u,v) = -(1+\epsilon)$.
\end{enumerate}
\end{proof}

We partition $W_v^0$ into two parts: $x$ and $y$. The first part denotes the contributions towards $W_v^0$ by the neighbors of $v$ that lie below level $k$, while the second part denotes the contribution towards $W_v^0$ by the neighbors of $v$ that lie on or above level $k$. Thus, we get the following equations.
\begin{eqnarray}
W_v^0 & = & x+ y \leq 1 \label{eq:part:1} \\
x  & = & \sum_{u \in \N_v(0,k-1)} w^0(u,v) = \beta^{-k} \cdot D_v(0,k-1) \label{eq:part:2} \\
y & = & \sum_{u \in \N_v(k,L)} w^0(u,v) \label{eq:part:3}
\end{eqnarray}

\begin{lemma}
\label{lm:FIX:case3:edge:sum}
We have $\sum_{u \in \N_v} \delta \Phi(u,v) =  -(1+\epsilon) \cdot x \cdot \beta^k$.
\end{lemma}

\begin{proof}
Lemma~\ref{lm:FIX:case3:edge} implies that $\sum_{u \in \N_v} \delta \Phi(u,v) = -(1+\epsilon) \cdot D_v(0,k-1)$. Applying equation~\ref{eq:part:2}, we infer that $D_v(0,k-1) =  x \cdot \beta^k$. The lemma follows.
\end{proof}

\begin{lemma}
\label{lm:new:1}
We have: 
$$\delta \Psi(v) = -\epsilon + (\alpha - x-y) \cdot  \left(\beta^{k+1}/(\beta-1)\right) - \max\left(0,\alpha - \beta x - y\right) \cdot \left(\beta^{k}/(\beta-1)\right).$$
\end{lemma}

\begin{proof}
Equation~\ref{eq:part:1} states that $W_v^0 = x+y < 1$. Since $\ell^0(v) = k$, we get:
\begin{equation}
\label{eq:FIX:case3:deltav:1}
\Psi^0(v) = \epsilon  (L-k) + (\alpha - x - y) \cdot \left(\beta^{k+1}/(\beta-1)\right)  
\end{equation}
As the node $v$ decreases its level from $k$ to $k-1$, we have: 
\begin{eqnarray*}
w^1(u,v)  = \begin{cases} 
\beta \cdot w^0(u,v) & \text{ if } u \in \N_v(0,k-1); \\ 
 w^0(u,v) & \text{ if } u \in \N_v(k,L)
\end{cases}
\end{eqnarray*} 
Accordingly, we have $W_v^1 = \beta \cdot x + y$, which implies the following equation.
\begin{equation}
\label{eq:FIX:case3:deltav:2}
\Psi^1(v) = \epsilon  (L-k+1) + \max(0,\alpha - \beta x -y) \cdot \left(\beta^{k}/(\beta-1)\right) 
\end{equation}
Since $\delta \Psi(v) = \Psi^0(v) - \Psi^1(v)$, the lemma now follows from equations~\ref{eq:FIX:case3:deltav:1} and~\ref{eq:FIX:case3:deltav:2}.
\end{proof}

\begin{proof}[Proof of Theorem~\ref{th:vc:analyze:FIX:main} (for Case 2)]
We consider two possible scenarios depending upon the value of $(\alpha - \beta x - y)$. We show that in each case  $\delta \B \geq \beta^k$. The theorem (for Case 2) then follows from Lemma~\ref{lm:case3:runtime}.
\begin{enumerate}
\item Suppose that $(\alpha - \beta x - y) < 0$. From  Lemmas~\ref{lm:FIX:case3:node:u},~\ref{lm:FIX:case3:edge:sum},~\ref{lm:new:1} and equation~\ref{eq:vc:change:1}, we derive:
\begin{eqnarray}
\epsilon \cdot \delta \B & = & \sum_{u \in \N_v} \delta \Psi(u) + \sum_{u \in \N_v} \delta \Phi(u,v) + \delta \Psi(v) \nonumber \\
& \geq & 0 - (1+\epsilon) \cdot x \cdot \beta^k - \epsilon + (\alpha - x - y) \cdot \beta^{k+1}/(\beta-1) \nonumber \\
& \geq & -\epsilon - (1+\epsilon)  \cdot \beta^k + (\alpha -1) \cdot \beta^{k+1}/(\beta-1) \qquad \qquad \qquad \qquad \qquad (\text{equation~\ref{eq:part:1}}) \nonumber \\
& = & -\epsilon + \frac{\beta^k}{(\beta-1)} \cdot \left\{ (\alpha-1) \cdot \beta - (1+\epsilon) (\beta-1) \right\} \nonumber \\
& = & -\epsilon + 2 \cdot (1+\epsilon) \cdot \beta^k \qquad \qquad \qquad \qquad \qquad \qquad \qquad (\text{since } \alpha = 1+3\epsilon \text{ and } \beta = 1+\epsilon) \nonumber \\
& \geq & \epsilon \cdot \beta^k \qquad \qquad \qquad \qquad \qquad \qquad \qquad \qquad \qquad  \qquad \qquad \qquad   (\text{since } \beta > 1, \epsilon \leq 1) \nonumber 
\end{eqnarray}
\item Suppose that $(\alpha - \beta x - y) \geq 0$. From  Lemmas~\ref{lm:FIX:case3:node:u},~\ref{lm:FIX:case3:edge:sum},~\ref{lm:new:1} and equation~\ref{eq:vc:change:1}, we derive:
\begin{eqnarray*}
\epsilon \cdot \delta \B & = & \sum_{u \in \N_v} \delta \Psi(u) + \sum_{u \in \N_v} \delta \Phi(u,v) + \delta \Psi(v) \\
& \geq & 0 - (1+\epsilon) \cdot x \cdot \beta^k - \epsilon + (\alpha - x - y) \cdot \beta^{k+1}/(\beta-1) -  (\alpha - \beta x - y) \cdot \beta^{k}/(\beta-1) \\
& = & -\epsilon +  \frac{\beta^k}{(\beta-1)} \cdot \big\{(\alpha - x - y) \cdot \beta - (1+\epsilon) \cdot x \cdot (\beta-1) - (\alpha -\beta x -y)\big\} \\
& = &  -\epsilon +  \frac{\beta^k}{(\beta-1)} \cdot \big\{(\alpha - x - y) \cdot (\beta - 1)   - \epsilon \cdot x \cdot (\beta -1)\big\} \\ 
& \geq &  -\epsilon +  \frac{\beta^k}{(\beta-1)} \cdot \big\{(\alpha - 1) \cdot (\beta -1) - \epsilon \cdot (\beta -1)\big\} \qquad \qquad  (\text{since }  0 \leq x+y \leq 1) \\
& = & -\epsilon + 2 \cdot \epsilon \cdot \beta^k    \qquad \qquad \qquad \qquad \qquad \qquad \qquad (\text{since } \alpha = 1+3\epsilon \text{ and } \beta = 1+\epsilon) \\
& \geq & \epsilon \cdot \beta^k  \qquad \qquad \qquad \qquad \qquad \qquad \qquad \qquad \qquad  \qquad \qquad \qquad   (\text{since } \beta > 1)
\end{eqnarray*}
\end{enumerate}

\end{proof}

\section{Dynamic Matching: Preliminaries}
\label{appendix:matching:1}

We are given an input graph $G = (V,E)$ that is being updated dynamically through a sequence of edge insertions/deletions. We want to maintain an approximately maximum matching in $G$. We will present three different algorithms for this problem. They are described in Sections~\ref{sec:matching:sqrtn},~\ref{sec:3matching} and~\ref{sec:4matching}. All these algorithms, however, will use two key ideas.
\begin{enumerate}
\item  It is easy to maintain a good approximate matching in a  bounded degree graph (see Section~\ref{sec:maxdeg}).
\item  Every graph contains a subgraph of bounded degree, called its {\em kernel}, that approximately preserves the size of the maximum matching (see Section~\ref{sec:invariant}).
\end{enumerate}
\noindent In Section~\ref{sec:preprocess}, we give a static algorithm for building a kernel in a graph. In Section~\ref{sec:data}, we present the data structures that will be used in Sections~\ref{sec:matching:sqrtn},~\ref{sec:3matching} and~\ref{sec:4matching} for maintaining a kernel in a dynamic setting. 

\paragraph{Query time.} In this paper, all the data structures for dynamic matching   explicitly maintain the set of matched edges. Accordingly, using appropriate pointers, we can support the following queries.
\begin{itemize}
\item Report the size of the matching $M$ maintained by the data structure in $O(1)$ time.
\item Report the edges in $M$ in $O(1)$ time per edge.
\item In $O(1)$ time, report whether or not a given edge $(u,v)$ is part of the matching.
\item In $O(1)$ time, report whether or not a given node $u$ is matched in $M$, and if yes, in $O(1)$ time report the node it is matched to.
\end{itemize}

\subsection{Maintaining an approximate matching in a bounded degree graph.}
\label{sec:maxdeg}

It is well known that maintaining a maximal matching is easy in a bounded degree graph. Still, for the sake of completeness,  we state this formally in the theorem below.

\begin{theorem}
\label{th:maxdeg:1}
Consider a dynamic graph $\G = (\V,\E)$, where the maximum degree of a node is always upper bounded by $\Delta$. Then we have an algorithm for maintaining a maximal matching $M$ in $\G$ that handles each edge insertion in $O(1)$ time and each edge deletion in $O(\Delta)$  time. 
\end{theorem}

\begin{proof}
When an edge $(u,v)$ is inserted into the graph $\G$, we check if both the nodes $u, v$ are unmatched in $M$. If yes, then we set $M \leftarrow M \cup \{(u,v)\}$. If no, then we do not change the matching $M$.

When an edge $(u,v) \in M$ is deleted from the graph $\G$, we first set $M \leftarrow M \setminus \{(u,v)\}$. Next, we  go through all the neighbors of $u$ to find out if any one of them is unmatched in $M$. If there  exists such a node $x$, with $(u,x) \in \E \setminus M$, then we set $M \leftarrow M \cup \{(u,x)\}$. Finally, we perform exactly the same operations on the node $v$.

When an edge $(u,v) \in \E \setminus M$ is deleted from the graph $\G$, we leave the matching $M$ as it is.

Clearly, this algorithm maintains a maximal matching, and handles each edge insertion (resp.~deletion) in $O(1)$ (resp.~$O(\Delta)$) time. 
\end{proof}

We now recall the definition of an augmenting path in a graph with respect to a given matching.

\begin{definition}
\label{def:app:path} Consider any graph $\G = (\V,\E)$ and any matching $M \subseteq \E$ in $\G$. For any nonnegative integer $k$, an {\em augmenting path}  of length $2k+1$ in $\G$, w.r.t.~$M$, is a path $(v_1, \ldots,  v_{2k+1})$ that satisfies the following properties.
\begin{itemize}
\item We have $(v_1, v_2) \in \E \setminus M$ and $(v_{2k},v_{2k+1}) \in \E \setminus M$.
\item For every $i \in \{1, \ldots, 2k-1\}$, if $(v_i,v_{i+1}) \in \E \setminus M$, then $(v_{i+1},v_{i+2}) \in M$. And  if $(v_i,v_{i+1}) \in  M$, then $(v_{i+1},v_{i+2}) \in \E \setminus M$.
\end{itemize}
\end{definition}

The following theorem is well known (see e.g.,~\cite{HopcroftK71}).

\begin{theorem}
\label{th:app:path}
If  $M$ is a matching in a graph $\G$ such that every augmenting path in $\G$ w.r.t.~$M$ has  length at least $(2k+1)$, then $M$ is a $(1+1/k)$-approximate maximum matching in $\G$.
\end{theorem}
Following the approach from~\cite{NeimanS13}, we show in the 
next theorem shows that in $O(\Delta)$ worst-case update time we can also maintain a $3/2$-approximate matching in a graph with maximum degree $\Delta$.

\begin{theorem}
\label{th:maxdeg:2}
Consider a dynamic graph $\G = (\V,\E)$ that is being  updated through a sequence of edge insertions/deletions. Furthermore, the maximum degree of a node in  $\G$ is always upper bounded by $\Delta$. Then we have a dynamic algorithm for maintaining a matching $M$ in $\G$ that has the following properties.
\begin{itemize}
\item The algorithm guarantees that every augmenting path in $\G$ (w.r.t.~$M$) has length at least five.
\item Each edge insertion into  $\G$ is handled in $O(\Delta)$ worst-case  time.
\item Each edge deletion  from $\G$ is handled in $O(\Delta)$ worst-case  time. 
\end{itemize}
\end{theorem}

\subsubsection{Proof of Theorem~\ref{th:maxdeg:2}}
\label{sec:maxdeg:2}

The algorithm maintains the following data structures for each node $v \in \V$.

\begin{itemize}
\item A doubly linked list $\text{{\sc Neighbors}}(v)$. It  consists of the set of nodes that are adjacent to $v$ in $\G$.
\item A doubly linked list $\text{{\sc Free-Neighbors}}(v)$. It consists of the unmatched (in $M$) neighbors of $v$.
\end{itemize}

\noindent Furthermore, the algorithm maintains two $|\V| \times |\V|$ matrices $Q_1$ and $Q_2$. For every pair of nodes $u, v$, if $u \in \text{{\sc Neighbors}}(v)$, then the entry $Q_1[u,v]$ stores a pointer to the position  of $u$ in the list  $\text{{\sc Neighbors}}(v)$. Similarly, if $u \in \text{{\sc Free-Neighbors}}(v)$, then the entry $Q_2[u,v]$ stores a pointer to the position of $u$ in the list $\text{{\sc Free-Neighbors}}(v)$. Thus,  a node $u$ can be inserted into (or deleted from) a list $\text{{\sc Neighbors}}(v)$ or $\text{{\sc Free-Neighbors}}(v)$ in constant time.

\paragraph{Handling an edge insertion.} Suppose that the edge $(u,v)$ is inserted into $\G$. We  make our data structures consistent with $\G$ by adding the node $u$ to the list $\text{{\sc Neighbors}}(v)$ and the node $v$ to the list $\text{{\sc Neighbors}}(u)$.

Before the insertion, every augmenting path in $\G$ (w.r.t $M$) had length at least five. After the insertion, some augmenting paths (containing the edge $(u,v)$) may be created that are of length one or three.  To resolve such paths, we  consider three possible scenarios.
\begin{itemize}
\item {\bf Case 1.} Both the endpoints $u$ and $v$ are unmatched in $M$. Here, we  set $M \leftarrow M \cup \{(u,v)\}$. 
\item {\bf Case 2.} One the endpoints (say $u$) is matched in $M$, while the other endpoint $v$ is unmatched in $M$. Let $x$ be the node $u$ is matched to, i.e., $(u,x) \in M$. If the list $\text{{\sc Free-Neighbors}}(x)$ is empty, then  we do nothing and terminate the procedure. Otherwise, let $y$ be a node that appears in $\text{{\sc Free-Neighbors}}(x)$. Clearly, $(v,u,x,y)$ is an augmenting path of length three. We resolve this path by setting $M \leftarrow \left(M \setminus \{(u,x)\}\right) \cup \{(u,v), (x,y)\}$. The nodes $v,y$ are now matched in $M$. To reflect this change, for every node $z \in \text{{\sc Neighbors}}(y)$, we remove $y$ from the list $\text{{\sc Free-Neighbors}}(z)$. Similarly, for every node $z \in \text{{\sc Neighbors}}(v)$, we remove $v$ from the list $\text{{\sc Free-Neighbors}}(v)$. 
\item {\bf Case 3.} Both the endpoints $u$ and $v$ are matched in $M$. In this case, the insertion of the edge $(u,v)$ does not create any augmenting path in $\G$ (w.r.t $M$) of length one or three. Hence, we do nothing and terminate the procedure.
\end{itemize}

\noindent If we are in Case 1 or in Case 3, then the procedure takes $O(1)$ time. In contrast, if we are in Case 2, then the procedure takes $O(\Delta)$ time. So the worst-case update time for handling an edge insertion is $O(\Delta)$.

\paragraph{Handling an edge deletion.} Suppose that the edge $(u,v)$ is deleted from the graph $\G$.  We first  make our data structures consistent with $\G$ by removing the node $u$ from the list $\text{{\sc Neighbors}}(v)$ and the node $v$ from the list $\text{{\sc Neighbors}}(u)$.  Next, we  consider two possible scenarios.
\begin{itemize}
\item {\bf Case 1.} $(u,v) \notin M$. In this case, no augmenting path of length one or three has been created.  Thus, we do nothing and terminate the procedure.
\item {\bf Case 2.} $(u,v) \in M$. We  set $M \leftarrow M \setminus \{(u,v)\}$. The nodes $u$ and $v$ are no longer matched in $M$. To reflect this change,  for every node $x \in \text{{\sc Neighbors}}(u)$, we add $u$ to the list $\text{{\sc Free-Neighbors}}(x)$. Similarly, for every node $x \in \text{{\sc Neighbors}}(v)$, we add $v$ to the list $\text{{\sc Free-Neighbors}}(x)$.  This may create some augmenting paths in $\G$ (w.r.t.~$M$) that are of length one or three. All such bad augmenting paths, however, start from either the node $u$ or the node $v$. To fix this problem, we call the subroutines $\text{{\sc Find-mate}}(u)$ and $\text{{\sc Find-mate}}(v)$, one after the other (see Figure~\ref{fig:resolve1}). 

If the set of unmatched neighbors of  $u$ (resp.~$v$) is nonempty, then $\text{{\sc Find-mate}}(u)$ (resp.~$\text{{\sc Find-mate}}(v)$) matches $u$ (resp.~$v$) to any one of them.  If the subroutine $\text{{\sc Find-mate}}(u)$ (resp.~$\text{{\sc Find-mate}}(v)$) succeeds in finding a mate for $u$ (resp.~$v$), then we know for sure that there is no more augmenting path in $\G$ of length one or three that starts from $u$ (resp.~$v$). In case of a failure, however, we have to deal with the possibility that there is an augmenting path of length three in the graph $\G$ that starts from $u$ (resp.~$v$).

 To take care of this issue, we finally call the subroutines $\text{{\sc Resolve}}(u)$ and $\text{{\sc Resolve}}(v)$, one after the other (see Figure~\ref{fig:resolve3}). The former (resp.~latter) subroutine resolves an augmenting path of length three, if there is any, starting from the node $u$ (resp.~$v$). At their termination, we are guaranteed that every augmenting path in $\G$ (w.r.t.~$M$) has length at least five.
\end{itemize}

\begin{figure}[htbp]
\centerline{\framebox{
\begin{minipage}{5.5in}
\begin{tabbing}
{\sc For all} $x \in \text{{\sc Neighbors}}(u)$: \\
\ \ \ \  \ \ \ \ \=  {\sc If} $x$ is unmatched in $M$, {\sc Then} \\
\> \ \ \ \ \ \  \ \ \ \ \= Set $M \leftarrow M \cup \{(u,x)\}$. \\
\> \>  {\sc For all} $z \in \text{{\sc Neighbors}}(u)$: \\
\> \> \ \ \ \ \  \ \  \ \ \ \  \= Remove $u$ from the list $\text{{\sc Free-Neighbors}}(z)$. \\
\> \>  {\sc For all} $z \in \text{{\sc Neighbors}}(x)$: \\
\> \> \> Remove $x$ from the list $\text{{\sc Free-Neighbors}}(z)$. 
\end{tabbing}
\end{minipage}
}}
\caption{\label{fig:resolve1} $\text{{\sc Find-mate}}(u)$.} 
\end{figure}

\begin{figure}[htbp]
\centerline{\framebox{
\begin{minipage}{5.5in}
\begin{tabbing}
{\sc If} $u$ is already matched in $M$, {\sc Then} \\
\ \ \ \ \ \ \ \ \ \  \= {\sc Return}. \\
 {\sc For all} $x \in \text{{\sc Neighbors}}(u)$: \\
  \>  We must have that $x$ is matched in $M$. \\
  \>   Let $y$ be the node $x$ is matched to, i.e., $(x,y) \in M$. \\
 \> {\sc If} $\text{{\sc Free-Neighbors}}(y) \setminus \{u\} \neq \emptyset$, {\sc Then} \\
   \> \ \ \ \ \ \ \ \ \ \= Consider any node $z \in \text{{\sc Free-Neighbors}}(y) \setminus \{u\}$. \\
   \> \>  Resolve the augmenting path $(z,y,x,u)$ in $\G$ (w.r.t.~$M$) by \\
   \>  \> setting $M \leftarrow \left(M \setminus \{(x,y)\}\right) \cup \{(y,z), (x,u)\}$. \\
   \>   \> {\sc For all} $w \in \text{{\sc Neighbors}}(u)$: \\
   \>  \> \ \ \ \ \ \ \ \ \ \  \= Remove $u$ from the list $\text{{\sc Free-Neighbors}}(w)$. \\
   \>   \> {\sc For all} $w \in \text{{\sc Neighbors}}(z)$: \\
   \>  \> \> Remove $z$ from the list $\text{{\sc Free-Neighbors}}(w)$. \\
 \> \> {\sc Return}.
\end{tabbing}
\end{minipage}
}}
\caption{\label{fig:resolve3} $\text{{\sc Resolve}}(u)$.} 
\end{figure}

\noindent If we are in Case 1, then the procedure runs for $O(1)$ time. In contrast, if we are in Case 2, then the procedure runs for $O(\Delta)$ time as both $\text{{\sc Find-mate}}$ and $\text{{\sc Resolve}}$ take time $O(\Delta)$. So the overall update time for handling an edge deletion is $O(\Delta)$.

\subsection{The kernel and its properties.}
\label{sec:invariant}

\newcommand{\F}{\mathcal{N}}


In the input  graph $G = (V,E)$, let $\N_v = \{u \in V : (u,v) \in E\}$ denote the set of {\em neighbors} of  $v \in V$. 

Consider a subgraph $\kappa(G) = (V, \kappa(E))$ with $\kappa(E) \subseteq E$. For all $v \in V$, define the set $\kappa(\N_v) = \{ u \in \N_v : (u,v) \in \kappa(E)\}$. If $u \in \kappa(\N_v)$, then we  say that  $u$ is a {\em friend} of  $v$ in $\kappa(G)$.  Next, the set of  nodes $V$ is partitioned into two groups: {\em tight} and {\em slack}. We denote the set of tight (resp.~slack) nodes  by $\kappa_{T}(V)$ (resp.~$\kappa_S(V)$). Thus, we have $V = \kappa_T(V) \cup \kappa_S(V)$ and $\kappa_T(V) \cap \kappa_S(V) = \emptyset$.

\begin{definition}
\label{def:kernel:appendix}
Fix any  $c \geq 1$ and any   $\epsilon \in [0,1/3)$.  The subgraph  $\kappa(G)$ is an $(\epsilon,c)$-kernel of   $G$ with respect to the partition $(\kappa_T(V), \kappa_S(V))$ iff it satisfies Invariants~\ref{inv:1}-~\ref{inv:3}.
\end{definition}

\begin{invariant}
\label{inv:1}
$|\kappa(\N_v)| \leq (1+\eps)c$ for all $v \in V$, i.e., every node has at most $(1+\eps)c$ friends. 
\end{invariant}

\begin{invariant}
\label{inv:2}
 $|\kappa(\N_v)| \geq (1-\eps)c$ for all $v \in \kappa_T(V)$, i.e., every tight node has at least $(1-\eps)c$ friends.
\end{invariant}

\begin{invariant}
\label{inv:3}
For all $u, v \in \kappa_S(V)$, if $(u,v) \in E$, then $(u,v) \in \kappa(E)$. In other words, if two slack nodes are connected by an edge in $G$, then  that edge must belong to $\kappa(G)$.
\end{invariant}

By Invariant~\ref{inv:1}, the maximum degree of a node in an $(\eps,c)$-kernel is $O(c)$. In Theorems~\ref{th:invariant} and~\ref{th:3matching:1}, we show that an $(\eps,c)$-kernel $\kappa(G)$ is basically a subgraph of $G$ with maximum degree $O(c)$ that approximately preserves the size of the maximum matching.

For ease of exposition, we  often refer to a kernel $\kappa(G)$ without explicitly mentioning the underlying partition $(\kappa_T(V),\kappa_S(V))$.  The proofs of Theorems~\ref{th:invariant},~\ref{th:3matching:1} appear in Sections~\ref{sec:th:invariant} and~\ref{sec:3matching:1} respectively.

\begin{theorem}
\label{th:invariant}
Let $M$ be a maximal matching in an $(\epsilon,c)$-kernel $\kappa(G)$. Then  $M$ is a $(4+6\eps)$-approximation to the maximum matching in  $G$.
\end{theorem}


\begin{theorem}
\label{th:3matching:1}
Let $M$ be a matching in an $(\epsilon,c)$-kernel $\kappa(G)$ such that every augmenting path in $\kappa(G)$ (w.r.t.~$M$) has length at least five. Then $M$ is a $(3+3\epsilon)$-approximation to the maximum matching in $G$. 
\end{theorem}

The next theorem shows that the approximation ratio derived in Theorem~\ref{th:invariant} is essentially tight. Its proof appears in Section~\ref{sec:th:kernel:3}.
\begin{theorem}
\label{th:kernel:3}
For every positive even integer $c$, there is a graph $G = (V,E)$, a $(1/c,c)$-kernel $\kappa(G) = (V,\kappa(E))$, and a maximal matching $M$ in $\kappa(G)$ such that the size of the maximum matching in $G$ is at least $4$ times the size of $M$.
\end{theorem}

\subsubsection{Proof of Theorem~\ref{th:invariant}.}
\label{sec:th:invariant}

Let $V(M) = \{v \in V : (u,v) \in M \text{ for some } u \in \N_v\}$ be the set of nodes that are matched in $M$, and let $F_T = \kappa_T(V) \setminus V(M)$  be the subset of tight nodes in $\kappa(G)$  that are free in $M$.

\begin{lemma}
\label{lem:bound:free}
Recall that $M$  is a  maximal matching in the $(\epsilon,c)$-kernel $\kappa(G)$. We have $|F_T| \leq (2+6\eps) \cdot |M|$.  
\end{lemma}

\begin{proof}
We will show that $|F_T| \leq (1+3\eps) \cdot |V(M)|$. Since $|V(M)| = 2 \cdot |M|$,  the lemma follows.

We  design a charging scheme where  each node in $F_T$ contributes one dollar  to a {\em global fund}. So the total amount of money in this fund is equal to $|F_T|$ dollars. Below, we demonstrate how to transfer this fund to the nodes in $V(M)$ so that each node $x \in V(M)$ receives at most $(1+3\eps)$ dollars. This  implies that $|F_T| \leq (1+3\eps) \cdot |V(M)|$. 

Since $M$ is a maximal matching in $\kappa(G)$, we must have  $\kappa(\N_v) \subseteq V(M)$ for all $v \in F_T$. For each node $v \in F_T$, we distribute its one dollar  equally among its friends, i.e., each node $x \in \kappa(\N_v)$ gets $1/|\kappa(\N_v)|$ dollars from $v$. Since $F_T \subseteq \kappa_T(V)$,  Invariant~\ref{inv:2} implies that  $1/|\kappa(\N_v)| \leq 1/\left((1-\eps)c\right)$ for all $v \in F_T$. In other words, a node in $V(M)$ receives at most $1/\left((1-\eps)c\right)$ dollars from each of its friends under this money-transfer scheme. But,  by Invariant~\ref{inv:1},  a node  can have at most $(1+\eps)c$ friends. So the total amount of money received by a node in $V(M)$ is at most $(1+\eps)c/\left((1-\eps)c\right)  \leq  (1+3\eps)$,  for  $\eps \in [0,1/3)$. 
\end{proof}

To continue with the proof of Theorem~\ref{th:invariant}, let $M^{o} \subseteq E$ be a maximum-cardinality matching in  $G = (V,E)$. Define $M^{o}_1 \subseteq M^{o}$ to be the subset of edges whose both endpoints are unmatched in $M$, and let $M^{o}_2 = M^{o} \setminus M^{o}_1$. 

Consider any edge $(u,v) \in M^o_1$. By definition, both the nodes $u, v$ are free in $M$. Since $M$ is a maximal matching in $\kappa(G)$, it follows that the edge $(u,v)$ is not part of the kernel, i.e., $(u,v) \neq \kappa(E)$. Hence, Invariant~\ref{inv:3} implies that either $u \notin \kappa_S(V)$ or $v \notin \kappa_S(V)$. Without any loss of generality, suppose that $u \notin \kappa_S(V)$. This means that the node $u$ is tight, and, furthermore, it is free in $M$.   We infer that every edge in $M^o_1$ is incident to at least one node from $F_T$, where $F_T \subseteq \kappa_T(V)$ is the subset of tight nodes that are free in $M$. Accordingly, we have $|M^o_1| \leq |F_T|$. Combining this inequality with Lemma~\ref{lem:bound:free}, we get: 
\begin{equation}
\label{eq:1}
|M^o_1| \leq (2+6\eps) \cdot |M|
\end{equation}

Next, every edge in $M^o_2$ has at least one  endpoint that is matched in $M$. Thus, $M$ is a maximal matching in the graph $G' = (V, M^o_2 \cup M)$. Since $M^o_2$ is also a  matching in  $G'$, we get:
\begin{equation}
\label{eq:2}
|M^o_2| \leq 2 \cdot |M|
\end{equation}

The theorem follows if we add equations~\ref{eq:1} and~\ref{eq:2}.

\subsubsection{Proof of Theorem~\ref{th:3matching:1}.}
\label{sec:3matching:1}

Let $V(M) = \{v \in V : (u,v) \in M \text{ for some } u \in \N_v\}$ be the set of nodes that are matched in $M$. For all nodes $v \in V(M)$, let $e_{M}(v)$ denote the edge in $M$ that is incident to $v$. Let $F_T = \kappa_T(V) \setminus V(M)$  be the subset of tight nodes in $\kappa(G)$  that are free in $M$.

\begin{lemma}
\label{lm:3matching:1:1} 
Recall that $M$ is a matching in the $(\epsilon,c)$-kernel $\kappa(G)$ such that every augmenting path in $\kappa(G)$ (w.r.t.~$M$) has length at least five. 
For every edge $(u,v) \in M$, we have $|\left(\kappa(\N_u) \cap F_T\right)  \cup \left(\kappa(\N_v) \cap F_T \right)|  \leq (1+\eps)c$.
\end{lemma}

\begin{proof}
Suppose that the lemma is false and we have $|\left(\kappa(\N_u) \cap F_T\right)  \cup \left(\kappa(\N_v) \cap F_T\right)|  > (1+\eps)c$ for some edge $(u,v) \in M$. As Invariant~\ref{inv:1} guarantees that $|\kappa(\N_u)| \leq (1+\eps)c$ and $|\kappa(\N_v)| \leq (1+\eps)c$,  there has to be a pair of distinct nodes $u', v'\in F_T$  such that $u'\in \kappa(\N_u)$ and $v'\in \kappa(\N_v)$. This means that the path $(u',u,v,v')$ is an augmenting path in $\kappa(G)$ (w.r.t.~$M$) and has length three. We reach a contradiction. 
\end{proof}

\begin{lemma}
\label{lm:3matching:1:2}
Recall that $M$ is  a matching in the $(\epsilon,c)$-kernel $\kappa(G)$ such that every augmenting path in $\kappa(G)$ (w.r.t.~$M$) has length at least five. 
We have $|F_T| \leq (1+3\eps) \cdot |M|$.  
\end{lemma}

\begin{proof}
We  design a charging scheme where  each node in $F_T$ contributes one dollar  to a {\em global fund}. So the total amount of money in this fund is equal to $|F_T|$ dollars. Below, we demonstrate how to transfer this fund to the edges in $M$ so that each edge $e \in M$ receives at most $(1+3\eps)$ dollars. This  implies that $|F_T| \leq (1+3\eps) \cdot |M|$. 

Since $M$ is a maximal matching in $\kappa(G)$, we must have  $\kappa(\N_v) \subseteq V(M)$ for all $v \in F_T$. For each node $v \in F_T$, we distribute its one dollar  equally among the matched edges incident to its friends, i.e., for each node $x \in \kappa(\N_v)$, the edge $e_{M}(x)$ gets $1/|\kappa(\N_v)|$ dollars from $v$. Since $F_T \subseteq \kappa_T(V)$,  Invariant~\ref{inv:2} implies that  $1/|\kappa(\N_v)| \leq 1/\left((1-\eps)c\right)$ for all $v \in F_T$. In other words, an edge $(x,y)  \in M$ receives at most $1/\left((1-\eps)c\right)$ dollars from each of the nodes $v \in (\kappa(\N_x) \cap F_T) \cup (\kappa(\N_y) \cap F_T)$ under this money-transfer scheme. Hence, Lemma~\ref{lm:3matching:1:1} implies that the total amount of money received by an edge $(x,y) \in M$ is at most $(1+\eps)c/((1-\eps)c) \leq (1+3\eps)$, for  $\eps \in [0,1/3)$.
\end{proof}

To continue with the proof of Theorem~\ref{th:3matching:1}, define $M^{o} \subseteq E$ to be a maximum-cardinality matching in  $G = (V,E)$. Furthermore, let $M^{o}_1 \subseteq M^{o}$  be the subset of edges whose both endpoints are unmatched in $M$, and let $M^{o}_2 = M^{o} \setminus M^{o}_1$. 

Consider any edge $(u,v) \in M^o_1$. By definition, both the nodes $u, v$ are free in $M$. Since $M$ is a maximal matching in $\kappa(G)$, it follows that the edge $(u,v)$ is not part of the kernel, i.e., $(u,v) \neq \kappa(E)$. Hence, Invariant~\ref{inv:3} implies that either $u \notin \kappa_S(V)$ or $v \notin \kappa_S(V)$. Without any loss of generality, suppose that $u \notin \kappa_S(V)$. This means that the node $u$ is tight, and, furthermore, it is free in $M$.   We infer that every edge in $M^o_1$ is incident to at least one node from $F_T$. Accordingly, we have $|M^o_1| \leq |F_T|$. Combining this inequality with Lemma~\ref{lm:3matching:1:2}, we get: 
\begin{equation}
\label{eq:3matching:1}
|M^o_1| \leq (1+3\eps) \cdot |M|
\end{equation}

Next, every edge in $M^o_2$ has at least one  endpoint that is matched in $M$. Thus, $M$ is a maximal matching in the graph $G' = (V, M^o_2 \cup M)$. Since $M^o_2$ is also a  matching in  $G'$, we get:
\begin{equation}
\label{eq:3matching:2}
|M^o_2| \leq 2 \cdot |M|
\end{equation}

The theorem follows if we add equations~\ref{eq:3matching:1} and~\ref{eq:3matching:2}.

\subsubsection{Proof of Theorem~\ref{th:kernel:3}}
\label{sec:th:kernel:3}

We define the set of nodes $U = \{u_1, \ldots, u_{c/2}\}$, $U'= \{u'_1, \ldots, u'_{c/2}\}$, $U^* = \{u^*_1, \ldots, u^*_c\}$, $X = \{x_1, \ldots, x_c\}$, $Y = \{y_1, \ldots, y_{c/2}\}$ and $Z = \{z_1, \ldots, z_{c/2}\}$. Next, we define the following sets of edges.
\begin{eqnarray*}
E(U,U') & = & \bigcup_{i=1}^{c/2} \{(u_i,u'_i)\} \\
E(U,U^*) & = & \bigcup_{u \in U, u^* \in U^*} \{(u,u^*)\} \\
E(U',U^*) & = & \bigcup_{u' \in U, u^* \in U^*} \{(u',u^*)\} \\
E(X,U^*) & = & \bigcup_{i=1}^c \{(x_i,u^*_i)\} \\
E(Y,U,U',Z) & = & \bigcup_{i=1}^{c/2} \{(y_i,u_i), (u'_i,z_i)\}
\end{eqnarray*}
 Next, we consider a graph $G = (V,E)$ defined on the node-set $V = U \cup U' \cup U^* \cup X \cup Y \cup Z$ and edge-set $E = E(U,U') \cup E(U,U^*) \cup E(U',U^*) \cup E(X,U^*) \cup E(Y,U,U', Z)$. See Figure~\ref{fig:picture}. Define the subset of edges $\kappa(E) = E(U,U') \cup E(U,U^*) \cup E(U',U^*)$. It is easy to check that the subgraph $\kappa(G) = (V,\kappa(E))$ is a $(1/c,c)$-kernel in $G$, with the set of tight nodes being $\kappa_T(V) = U \cup U' \cup U^*$, and the set of slack nodes being $\kappa_S(V) = X \cup Y \cup Z$. We now consider two matchings $M$ and $M^*$ defines as follows. 
\begin{eqnarray*}
M & = & \bigcup_{i=1}^{c/2} \{ (u_i,u'_i) \}  \\
M^* & = & E(X,U^*) \bigcup E(Y,U,U',Z)
\end{eqnarray*}
The theorem follows from the observations that $M$ is a maximal matching in $\kappa(G)$ and  $|M| = c/2$, $|M^*| = 2c$.

\begin{figure}[!ht]
  \centering
      \includegraphics[width=0.6\textwidth]{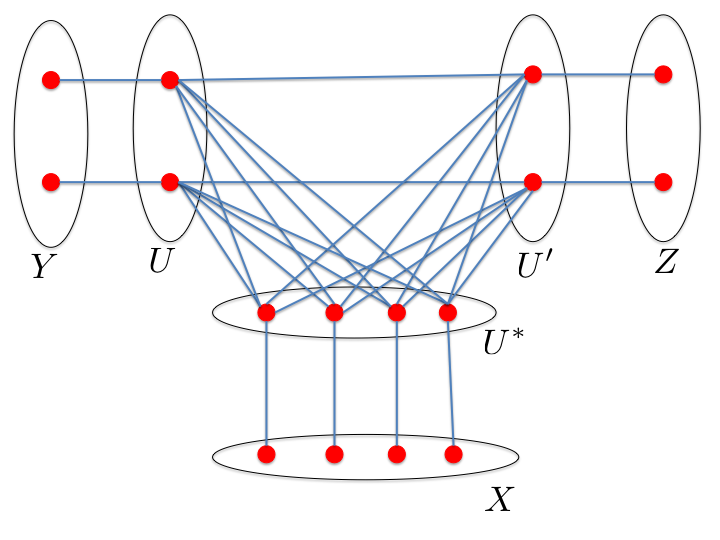}
  \caption{\label{fig:picture} The graph considered in the proof of Theorem~\ref{th:kernel:3}, for $c = 4$.}
  \end{figure}

\subsection{An algorithm for building a kernel in a static graph.}
\label{sec:preprocess}

We present a linear-time algorithm for constructing a $(0,c)$-kernel $\kappa(G)$ of a static graph $G = (V,E)$.

\begin{theorem}
\label{th:preprocess}
We have an algorithm for computing a $(0,c)$-kernel $\kappa(G) = (V,\kappa(E))$ of   a graph $G = (V,E)$. The kernel returned by the algorithm has the added property that 
tight nodes have exactly
$c$ friends and slack nodes have less than $c$ friends.
For every $c \geq 1$, the algorithm runs in $O(|E|)$ time.
\end{theorem}

\begin{proof}
Initially, each node in $V$ has zero friends, and the edge-set $\kappa(E)$ is empty. We then execute the {\sc For} loop stated below.
\begin{itemize}
\item {\sc For all} $(u,v) \in E$
\begin{itemize}
\item {\sc if} $|\kappa(\N_u)| < c$ and $|\kappa(\N_v)| < c$, {\sc Then}
\begin{itemize} 
\item Set $\kappa(\N_v) \leftarrow \kappa(\N_v) \cup \{u\}$ and $\kappa(\N_u) \leftarrow \kappa(\N_u) \cup \{v\}$. 
\item Set $\kappa(E) \leftarrow \kappa(E) \cup \{(u,v)\}$. 
\end{itemize}
\end{itemize}
\end{itemize}

\noindent Consider the  $\kappa(G) = (V,\kappa(E))$ we get at the end of the {\sc For} loop. Clearly, in $\kappa(G)$  every  node has at most $c$ friends, i.e., $|\kappa(\N_v)| \leq c$  for all $v \in V$. Furthermore, for all edges $(u,v) \in E$, if $|\kappa(\N_u)| < c$ and $|\kappa(\N_v)| < c$, then it is guaranteed that $(u,v) \in \kappa(E)$. Thus, if we define $\kappa_T(V) = \{ v \in V : |\kappa(\N_v)| = c\}$ and $\kappa_S(V) =  \{ v \in V : |\kappa(\N_v)| < c\}$, then  $\kappa(G)$ becomes a $(0,c)$-kernel of $G$. The whole procedure can be implemented in $O(|E|)$ time.
\end{proof}

\subsection{Data structures for representing a kernel in a dynamic graph.}
\label{sec:data}

In the graph $G$, the set of   neighbors of a node $v$ is  stored in the form of a linked list $\text{{\sc neighbors}}(v)$, which is part of an adjacency-list data structure. Further, each node $v$ maintains the following information. 
\begin{itemize}
\item A bit $\text{{\sc Type}}(v)$  indicating if the node is tight or slack. 
\[
 \text{{\sc Type}}(v)  = \begin{cases}
  \text{tight} &  \text{if } v \in \kappa_T(V); \\
  \text{slack} & \text{if } v \in \kappa_S(V).
 \end{cases}
\]
\item The set of nodes $\kappa(\N_v)$, in the form of a doubly linked list $\text{{\sc Friends}}(v)$. 
\item A counter $\text{{\sc \#friends}}(v) = |\kappa(\N_v)|$ that keeps track of the number of  friends  of the node.
\end{itemize}


\noindent Furthermore, 
we store at each edge 
$(u,v) \in \kappa(E)$ two pointers, corresponding to the two occurrences of edge $(u,v)$ in the linked lists
$\text{{\sc Friends}}(u)$ and $\text{{\sc Friends}}(v)$. In particular, we denote by $\text{{\sc Pointer}}[u,v]$ (resp.~$\text{{\sc Pointer}}[v,u]$) the pointer to the position of $u$ (resp.~$v$) in the list $\text{{\sc Friends}}(v)$ (resp.~$\text{{\sc Friends}}(u)$).  Using those pointers,  we can insert an edge into $\kappa(G)$ or delete an edge from $\kappa(G)$ in $O(1)$ time.

\section{$(3+\eps)$-approximate matching in $O(\sqrt{n}/\eps)$ amortized update time}
\label{sec:matching:sqrtn}

Fix any $\eps \in (0,1/3)$. We present an algorithm for maintaining a $(3+\eps)$-approximate matching in a graph $G = (V,E)$ 
undergoing a sequence of updates with $O(\sqrt{n}/\eps)$ amortized update time, where $n = |V|$ denotes the number of nodes in $G$.

\subsection{Overview of our approach.}
\label{sec:matching:sqrtn:overview}

Before the graph $G = (V,E)$ starts changing, we preprocess it in linear time  to build a $(0,\sqrt{n})$-kernel $\kappa(G)$ of $G$ (see Theorem~\ref{th:preprocess}). Next, we compute a matching $M$ in $\kappa(G)$ such that every augmenting path in $\kappa(G)$ (w.r.t.~$M$) has length at least five. This also takes $O(|E|)$ time.

Thus, initially we have a $(0,\sqrt{n})$-kernel of $G$, which is also an $(\eps,\sqrt{n})$-kernel of $G$ for all $\eps \geq 0$. After each update (edge insertion/deletion) in $G$, we first modify  $\kappa(G)$ using the algorithm in Section~\ref{sec:n:kernel}, and subsequently we modify the matching $M$ in $\kappa(G)$ as per Theorem~\ref{th:maxdeg:2}. The algorithm in Section~\ref{sec:n:kernel} guarantees that $\kappa(G)$ remains an $(\eps,\sqrt{n})$-kernel of $G$ after this update. Hence, by Theorem~\ref{th:3matching:1},  the matching $M$ always remains a $(3+\eps)$-approximation to the maximum matching in $G$.

We now explain the idea behind the algorithm in Section~\ref{sec:n:kernel}. In the beginning of the phase, Theorem~\ref{th:preprocess} guarantees that every tight node has exactly $\sqrt{n}$ friends. During the phase, whenever the number of friends of a tight node $v$ drops below $(1-\eps)\sqrt{n}$, it   scans through its {\em entire list} of neighbors in $\N_v$. While passing through a node $u \in \N_v$ during this scan, if it discovers that both $|\kappa(\N_u)|, |\kappa(\N_v)| < \sqrt{n}$ and $u \notin \kappa(\N_v)$, then it adds the edge $(u,v)$ to $\kappa(G)$. This procedure is implemented in the subroutine $\text{REFILL-NOW-FRIENDS}(v)$ (see Figure~\ref{fig:adjust:n}). Note that this subroutine can take $O(n)$ time, but it is called only when a tight node has lost $\eps \sqrt{n}$ of its friends. Each such loss corresponds to an edge deletion in $\kappa(E) \subseteq E$. Thus, distributing this cost among the concerned edge deletions in $G$, we get an amortized update time of $O(\sqrt{n}/\eps)$.

Finally, in Section~\ref{sec:n:matching}, we show that the amortized update time for maintaining the matching $M$ in $\kappa(G)$ is also $O(\sqrt{n}/\eps)$, and this leads  to the main result of this section.

\begin{theorem}
\label{th:n:matching}
In a dynamic graph $G = (V,E)$, we can maintain a $(3+\eps)$-approximate matching $M \subseteq E$ in $O(\sqrt{n}/\eps)$ amortized update time.
\end{theorem}

\subsection{Maintaining the kernel $\kappa(G)$.}
\label{sec:n:kernel}

We will ensure that the kernel $\kappa(G)$ always satisfies the five properties described below. Clearly, these properties hold immediately after the preprocessing step (see Theorem~\ref{th:preprocess}).

\begin{itemize}
\item (P1) Every node has at most $\sqrt{n}$ friends in $\kappa(G)$. 
\item (P2) Every tight node has at least $(1-\eps)\sqrt{n}$ friends in $\kappa(G)$.
\item (P3) If two slack nodes are connected by an edge in $G$, then they are friends of each other in $\kappa(G)$.
\item (P4) If a  node has exactly $\sqrt{n}$ friends in $\kappa(G)$, then it is tight.
\item (P5) A slack node $v$ becomes tight only when $|\kappa(\N_v)|$ increases from $\sqrt{n}-1$ to $\sqrt{n}$.
\end{itemize}
\noindent  We now describe our algorithm for updating $\kappa(G)$ after an edge insertion/deletion in $G$. It is easy to check that if (P1)-(P5) were true before the update, then they continue to be true after the update. 

The  three properties (P1)-(P3) imply that $\kappa(G)$ is an $(\eps,\sqrt{n})$-kernel. The remaining properties (P4)-(P5) are helpful for reasoning about the correctness of our algorithm and for bounding its amortized update time.

\paragraph{Handling an edge insertion in $G$.}

Suppose that an edge $(u,v)$ is inserted into the graph $G = (V, E)$. To handle this edge insertion, we first  update the lists $\text{{\sc neighbors}}(u)$ and $\text{{\sc neighbors}}(v)$, and then process the edge as follows.

\begin{itemize}
\item {\bf Case 1:} Either  $\text{{\sc \#Friends}}(u) = \sqrt{n}$ or $\text{{\sc \#Friends}}(v) =  \sqrt{n}$.

We do nothing.

\item {\bf Case 2:} Both $\text{{\sc \#Friends}}(u) <  \sqrt{n}$ and $\text{{\sc \#Friends}}(v) <  \sqrt{n}$.

We add the edge $(u,v)$ to the kernel $\kappa(G)$ and make $u, v$ friends of each other. Specifically, we  add the node $u$ to the set $\text{{\sc Friends}}(v)$, the node $v$ to the set  $\text{{\sc Friends}}(u)$, update the pointers $\text{{\sc Pointer}}[u,v]$, $\text{{\sc Pointer}}[v,u]$, and increment each of  the counters $\text{{\sc \#Friends}}(u)$, $\text{{\sc \#Friends}}(v)$ by one unit. After this increment, if $\text{{\sc \#Friends}}(u) =  \sqrt{n}$, then we set $\text{{\sc Type}}(u) \leftarrow \text{tight}$. Similarly, if $\text{{\sc \#Friends}}(v) =  \sqrt{n}$, then we set $\text{{\sc Type}}(v) \leftarrow \text{tight}$.
\end{itemize}

\paragraph{Handling an edge deletion in $G$.}

Suppose that an edge $(u,v)$ gets deleted from the graph $G = (V, E)$. To handle this edge deletion, we first  update the lists $\text{{\sc neighbors}}(u)$ and $\text{{\sc neighbors}}(v)$.

Next, we  check if the edge $(u,v)$ was part of the kernel  $\kappa(G)$, and, if yes, then we delete $(u,v)$ from $\kappa(E)$.  Specifically,  we delete $u$ from $\text{{\sc Friends}}(v)$, $v$ from $\text{{\sc Friends}}(u)$, update the pointers $\text{{\sc Pointer}}[u,v]$, $\text{{\sc Pointer}}[v,u]$, and decrement each of the counters $\text{{\sc \#Friends}}(u)$, $\text{{\sc \#Friends}}(v)$ by one unit.  

We then process the nodes $u$ and $v$ one after another. Below, we describe  only the procedure that runs on the node $u$. The procedure for the node $v$ is exactly the same.

\begin{itemize}
\item {\bf Case 1:} $\text{{\sc Type}}(u) = \text{tight}$. Here, we  check if the number of friends of $u$ has dropped below the prescribed limit  due to the edge deletion, and, accordingly, we  consider two possible subcases.
\begin{itemize}
\item {\bf Case 1a:} $\text{{\sc \#Friends}}(u) < (1- \epsilon) \sqrt{n}$. Here,  we call the subroutine $\text{{\sc REFILL-NOW-FRIENDS}}(u)$  described in Figure~\ref{fig:adjust:n}. This takes $O(n)$ time.
\item {\bf Case 1b:} $\text{{\sc \#Friends}}(u) \geq (1 -\epsilon) \sqrt{n}$. We do nothing in this case.
\end{itemize}
\item {\bf Case 2:} $\text{{\sc Type}}(u) = \text{slack}$. We do nothing in this case. 
\end{itemize}

\begin{figure}[htbp]
\centerline{\framebox{
\begin{minipage}{5.5in}
\begin{tabbing}
{\sc If} the list $\text{{\sc neighbors}}(u)$ is empty, {\sc Then} \\
\ \ \ \ \ \ \ \ \  \=  Set $\text{{\sc Type}}(u) \leftarrow \text{slack}$. \\
\> {\sc  Return}. \\
  Let $x$ be the first node in the list $\text{{\sc neighbors}}(u)$. \\ 
  {\sc While} $\left(\text{{\sc \#Friends}}(u) <  \sqrt{n}\right)$ \\
  \> {\sc If} $x \notin \text{{\sc Friends}}(u)$ and $\text{{\sc \#Friends}}(x) <  \sqrt{n}$, {\sc Then} \\
   \> \ \ \ \ \ \= Insert  $x$ into the list $\text{{\sc Friends}}(u)$, and  $u$ into the list $\text{{\sc Friends}}(x)$, so that  the edge $(u,x)$ \\
   \> \> becomes part of the kernel $\kappa(G) = (V,\kappa(E))$. Update the pointers  $\text{{\sc Pointer}}[u,x]$ and \\
\> \>   $\text{{\sc Pointer}}[x,u]$. Increment each of the counters  $\text{{\sc \#Friends}}(u)$, $\text{{\sc \#Friends}}(v)$   by one unit. \\
\> \>  After this increment, if $\text{{\sc \#Friends}}(x) =  \sqrt{n}$, then set $\text{{\sc Type}}(x) \leftarrow \text{tight}$. \\
 \>  {\sc If} $x$ is the last node in the list $\text{{\sc neighbors}}(u)$, {\sc Then}  \\ 
 \> \> exit the {\sc While} loop. \\
  \>  {\sc Else} \\
\>    \> Let $y$ be the node that succeeds $x$ in the list $\text{{\sc neighbors}}(u)$. \\
  \>  \>  Set $x \leftarrow y$. \\ 
 {\sc If} $\text{{\sc \#Friends}}(u) <  \sqrt{n}$, {\sc Then} \\
 \> Set $\text{{\sc Type}}(u) \leftarrow \text{slack}$.
  \end{tabbing}
\end{minipage}
}}
\caption{\label{fig:adjust:n} $\text{{\sc REFILL-NOW-FRIENDS}}(u)$.} 
\end{figure}

\paragraph{Bounding the amortized update time.} If we ignore the calls to the subroutine $\text{REFILL-NOW-FRIENDS}(.)$, then handling every edge insertion/deletion in $G$ requires $O(1)$ time in the worst case.  We will show that the time taken by all the calls to $\text{REFILL-NOW-FRIENDS}(.)$, amortized over a sequence of updates in $G$, is $O(\sqrt{n}/\eps)$.

Fix any sequence of $t$ edge updates in $G$, and let $t_v$ denote the number of edge updates in this sequence that are incident to a node $v \in V$. Since each edge is incident to two nodes, we have $2t = \sum_{v \in V} t_v$. Let $r_v$ denote the number of times the subroutine $\text{REFILL-NOW-FRIENDS}(v)$ is called during this sequence, so that the total number of calls to the subroutine $\text{REFILL-NOW-FRIENDS}(.)$ is given by $\sum_{v \in V} r_v$. Since each individual call to $\text{REFILL-NOW-FRIENDS}(.)$ takes $O(n)$ time in the worst case, it suffices to show the following guarantee.
\begin{equation}
\label{eq:n:runtime}
\frac{\sum_v r_v}{t} \leq \frac{\alpha}{\eps \sqrt{n}}, \text{ for some constant } \alpha > 0.
\end{equation} 

Consider any node $v \in V$. Our preprocessing step guarantees that in the beginning of the algorithm, the node $v$ is tight only if it has exactly $\sqrt{n}$ friends in $\kappa(G)$. By (P5), during the sequence of updates in $G$, the node $v$ changes its status from slack to tight only if $|\kappa(\N_v)|$ becomes equal to $\sqrt{n}$. On the other hand, the subroutine $\text{REFILL-NOW-FRIENDS}(v)$ is called only if $v$ is tight and $|\kappa(\N_v)|$ drops below $(1-\eps) \sqrt{n}$. Accordingly, a call to $\text{REFILL-NOW-FRIENDS}(v)$ occurs when $v$ has lost at least $\eps \sqrt{n}$ friends in $\kappa(G)$, and each such loss stems from  the deletion of an edge in $\kappa(E) \subseteq E$ that was incident to $v$. Thus, we conclude that $r_v/t_v \leq 1/(\eps \sqrt{n})$. We now derive equation~\ref{eq:n:runtime} as follows.
$$\frac{\sum_v r_v}{t} = 2 \cdot \frac{\sum_v r_v}{\sum_v t_v} \leq 2 \cdot \max_v \left(\frac{r_v}{t_v}\right) = \frac{2}{\eps \sqrt{n}}.$$

\begin{lemma}
\label{lm:n:runtime}
We can maintain an $(\eps,c)$-kernel $\kappa(G)$ of a dynamic graph $G$. Updating the kernel  requires $O(\sqrt{n}/\eps)$ time, amortized over the number of edge insertions/deletions in $G$.
\end{lemma}

\subsection{Maintaining the matching $M$ in $\kappa(G)$.}
\label{sec:n:matching}

We maintain the matching $M$ in $\kappa(G)$ using Theorem~\ref{th:maxdeg:2}. To bound the amortized update time of this procedure, fix any sequence of $t$ edge updates in $G$.  Ignoring the calls to  $\text{REFILL-NOW-FRIENDS}(.)$, each update in $G$ leads to at most one update in $\kappa(G)$. On the other hand, every single call to $\text{REFILL-NOW-FRIENDS}(.)$ leads to $O(\sqrt{n})$ updates in $\kappa(G)$. By equation~\ref{eq:n:runtime},  a sequence of   $t$ updates in $G$ results in at most $O(t/(\eps \sqrt{n}))$ calls to $\text{REFILL-NOW-FRIENDS}(.)$. Hence, such a sequence can cause at most $O(t/\eps)$ updates in $\kappa(G)$. As the maximum degree of a node in $\kappa(G)$ is $O(c)$, each update in $\kappa(G)$ is handled in $O(\sqrt{n})$ time by the algorithm in Theorem~\ref{th:maxdeg:2}.  This gives us a total running time  of $O((t/\eps) \cdot \sqrt{n})$ for maintaining the matching $M$ during a sequence of $t$ updates in $G$. Thus, the amortized update time becomes $O(\sqrt{n}/\eps)$.

\begin{lemma}
\label{lm:n:runtime:1}
 In a dynamic graph $G$, we can maintain  an $(\eps,c)$-kernel $\kappa(G)$ and a matching $M$ in $\kappa(G)$ such that every augmenting path in $\kappa(G)$ (w.r.t.~$M$) has length at least five. Updating the kernel $\kappa(G)$ and the matching $M$ requires $O(\sqrt{n}/\eps)$ time, amortized over the number of egde insertions/deletions in $G$.
\end{lemma}

\begin{proof}[Proof of Theorem~\ref{th:n:matching}] By Theorem~\ref{th:3matching:1}, the matching maintained by Lemma~\ref{lm:n:runtime:1} is a $(3+\epsilon)$-approximate maximum matching in the input graph $G = (V, E)$.

\end{proof}

\section{$(3+\eps)$-approximate matching in $O(m^{1/3}/\eps^2)$ amortized update time}
\label{sec:3matching}
\label{app:matching:2}

Fix any constant $\eps \in (0,1/3)$. In this section, we present an algorithm for maintaining a $(3+\eps)$-approximate matching $M$ in a graph $G = (V,E)$ undergoing a sequence of edge insertions/deletions. It requires $O(m^{1/3}/\eps^2)$ amortized update time.

\subsection{Overview of our approach.}
\label{sec:3matching:overview}

The main idea is to partition the sequence of updates (edge insertions/deletions) in $G$ into {\em phases}. Each phase lasts for $\eps^2 c^2/2$ consecutive updates in $G$. Let $G_{i,t}$ denote the state of  $G$ just after the $t^{th}$ update in phase $i$, for any $0 \leq t \leq \eps^2 c^2/2$. The initial state of the graph, before it starts changing, is given by $G_{1,0}$. We reach the graph $G_{i,t}$ from $G_{1,0}$ after a sequence of $(i-1) \cdot (\eps^2 c^2/2) + t$ updates in $G$.

For the rest of this section, we  focus on describing our algorithm for any given phase $i \geq 1$. We define $m \leftarrow |E_{i,0}|$ to be the number of edges in the input graph in the beginning of the phase, and  set $c = m^{1/3}$. Since the phase lasts for only $O(\eps^2 m^{2/3})$ updates in $G$, it follows that $|E_{i,t}| = O(m)$ for all $0 \le t \le \epsilon^2 c^2/2$.  During the phase, we   maintain a $(3+\eps)$-approximate matching in $G$  as described below.

Just before phase $i$ begins, we build a $(0,c)$-kernel $\kappa(G)$ on $G = G_{i,0}$ as per Theorem~\ref{th:preprocess}. This takes $O(m)$ time. Next, we compute a $3/2$-approximate maximum matching $M = M_{i,0}$ on $\kappa(G_{i,0})$ by ensuring that every augmenting path in $\kappa(G_{i,0})$ (w.r.t.~$M_{i,0}$) has length at least five. This also takes $O(m)$ time, and concludes the preprocessing step.

After each update in $G$ during phase $i$, we first modify the kernel $\kappa(G)$ using the algorithm  in Section~\ref{sec:3matching:maintain:1}, and subsequently we modify the matching $M$ in $\kappa(G)$ using Theorem~\ref{th:maxdeg:2}. Theorem~\ref{th:phase:new} below guarantees that the graph $\kappa(G)$ remains an $(\eps, c)$-kernel of $G$ throughout the phase. Hence, by Theorem~\ref{th:3matching:1}, the matching $M$ remains a $(3+\eps)$-approximation to the maximum matching in $G$.

The idea behind the algorithm in Section~\ref{sec:3matching:maintain:1} is simple. In the beginning of the phase, Theorem~\ref{th:preprocess} guarantees that every tight node has exactly $c$ friends. During the phase, whenever the number of friends of a tight node $v$ drops below $(1-\eps)c$, it   scans through its first $c$ neighbors in $\N_v$, and keeps making friends out of them until $|\kappa(\N_v)|$ becomes equal to $c$. This procedure is implemented in the subroutine $\text{REFILL}(v)$ (see Figure~\ref{fig:adjust:new}). A potential problem with this approach is that after a while a node $v$ may have more than $(1+\eps)c$ friends due to the repeated invocations of the subroutine $\text{REFILL}(u)$, for neighbors $u \in \N_v$. We show that this event can be ruled out  by ending the current phase after $\eps^2 c^2/2$ updates in $G$.

\paragraph{Bounding the update time.} 
\begin{itemize}
\item {\em Preprocessing.}

The preprocessing in the beginning of the phase takes $O(m)$ time. Since the phase lasts for $\eps^2 c^2/2$ updates, we get an amortized bound of $O(m/(\eps^2 c^2)) = O(m^{1/3}/\eps^2)$.

\item {\em Maintaining the kernel $\kappa(G)$.} 

By Theorem~\ref{th:phase:new} below, the kernel $\kappa(G)$ can be modified after each update in $G$ in $O(c) = O(m^{1/3})$ time. 

\item {\em Maintaining the matching $M$ in $\kappa(G)$.}

By Theorem~\ref{th:phase:new} below, the number of updates made into $\kappa(G)$ during the entire phase is $O(\eps^2 c^2)$. Since the maximum degree of a node in $\kappa(G)$ is $O(c)$, modifying the matching $M$ after each update in $\kappa(G)$ requires $O(c)$ time (see Theorem~\ref{th:maxdeg:2}). Thus, the total time spent during the phase in maintaining the matching $M$ is $O(\eps^2 c^2) \cdot O(c) = O(\eps^2 c^3)$. Since the phase lasts for $\eps^2 c^2/2$ updates in $G$, we get an amortized bound of $O(c) = O(m^{1/3})$.
\end{itemize}
\noindent We summarize the main result of Section~\ref{sec:3matching} in the theorem below. Note that in order to prove Theorem~\ref{th:3matching:main:appendix}, we only need to show how to maintain an $(\epsilon, c)$-kernel  in $O(c)$ update time. We do this in Section~\ref{sec:3matching:maintain:1}.

\begin{theorem}
\label{th:3matching:main:appendix}
In a dynamic graph $G = (V,E)$, we can maintain a $(3+\eps)$-approximate matching $M \subseteq E$ in $O(m^{1/3}/\eps^2)$ amortized update time.
\end{theorem}

\subsection{Algorithms for maintaining the kernel during a phase.}
\label{sec:3matching:maintain:1}

We present our algorithm for maintaining an $(\eps,c)$-kernel of the graph during a phase. 

\begin{theorem}
\label{th:phase:new}
Suppose that we are given a $(0,c)$-kernel $\kappa(G)$ of $G$  in the beginning of a phase. Then we have an algorithm for modifying $\kappa(G)$ after each update (edge insertion/deletion) in $G$ such that:
\begin{enumerate}
\item Each   update in $G$ is handled in $O(c)$ worst-case  time.
\item At most  $O(\eps^2 c^2)$ updates  are made into $\kappa(G)$ during the entire phase.
\item Throughout the duration of the phase,  $\kappa(G)$ remains an $(\eps,c)$-kernel of $G$.
\end{enumerate}
\end{theorem}

\noindent The rest of Section~\ref{sec:3matching:maintain:1} is organized as follows. In Section~\ref{sec:insertion:new} (resp.~\ref{sec:deletion:new}), we present our   algorithm  for handling an edge insertion  (resp.~edge deletion) in $G$ during the phase. In Section~\ref{sec:phase:proof:new}, we  prove that the algorithm satisfies the properties stated in Theorem~\ref{th:phase:new}.

\subsubsection{Handling an edge insertion in $G$ during the phase.}
\label{sec:insertion:new}

Suppose that an edge $(u,v)$ is inserted into the graph $G = (V, E)$. To handle this edge insertion, we first  update the lists $\text{{\sc neighbors}}(u)$ and $\text{{\sc neighbors}}(v)$, and then process the edge as follows.

\begin{itemize}
\item {\bf Case 1:} Either $\text{{\sc Type}}(u)  = \text{tight}$ or $\text{{\sc Type}}(v)  = \text{tight}$.

 We do nothing and conclude the procedure.
\item {\bf Case 2:} Both $\text{{\sc Type}}(u) = \text{slack}$ and $\text{{\sc Type}}(v) = \text{slack}$.
\begin{itemize}
\item {\bf Case 2a:} Either $\text{{\sc \#Friends}}(u) \geq c$ or $\text{{\sc \#Friends}}(v) \geq c$.

If $\text{{\sc \#Friends}}(u) \geq c$, then we set $\text{{\sc Type}}(u) \leftarrow \text{tight}$. Next, if $\text{{\sc \#Friends}}(v) \geq c$, then we set $\text{{\sc Type}}(v) \leftarrow \text{tight}$.
\item {\bf Case 2b:} Both $\text{{\sc \#Friends}}(u) < c$ and $\text{{\sc \#Friends}}(v) < c$. 

We add the edge $(u,v)$ to the kernel $\kappa(G)$ and make $u, v$ friends of each other. Specifically, we  add the node $u$ to the list $\text{{\sc Friends}}(v)$ and  $v$ to the list $\text{{\sc Friends}}(u)$, update the pointers $\text{{\sc Pointer}}[u,v]$ and $\text{{\sc Pointer}}[v,u]$ accordingly, and increment each of  the counters $\text{{\sc \#Friends}}(u)$, $\text{{\sc \#Friends}}(v)$ by one unit.
\end{itemize}
\end{itemize}

\begin{lemma}
\label{lm:insertion:new}
Suppose that an edge is inserted into the  graph $G$ and we run the procedure described above.
\begin{enumerate}
\item This causes at most one edge insertion into $\kappa(G)$ and at most one edge deletion from   $\kappa(G)$.
\item  The procedure runs in $O(1)$ time in the worst case.
\end{enumerate}
\end{lemma}

\subsubsection{Handling an edge deletion in $G$ during the phase.}
\label{sec:deletion:new}

Suppose that an edge $(u,v)$ is deleted from the graph $G = (V, E)$. To handle this edge deletion, we first  update the lists $\text{{\sc neighbors}}(u)$ and $\text{{\sc neighbors}}(v)$, and then proceed as follows.

We  check if the edge $(u,v)$ was part of the kernel  $\kappa(G)$, and, if yes, then we delete $(u,v)$ from $\kappa(G)$.  Specifically,  we delete $u$ from $\text{{\sc Friends}}(v)$, $v$ from $\text{{\sc Friends}}(u)$ (using $\text{{\sc Pointer}}[u,v]$ and $\text{{\sc Pointer}}[v,u]$), and decrement each of the counters $\text{{\sc \#Friends}}(u)$, $\text{{\sc \#Friends}}(v)$ by one unit.

We then process the nodes $u$ and $v$ one after another. Below, we describe  only the procedure that runs on the node $u$. The procedure for the node $v$ is exactly the same.

\begin{itemize}
\item {\bf Case 1:} $\text{{\sc Type}}(u) = \text{tight}$. Here, we  check if the number of friends of $u$ has dropped below the prescribed limit  due to the edge deletion, and, accordingly, we  consider two possible sub-cases.
\begin{itemize}
\item {\bf Case 1a:} $\text{{\sc \#Friends}}(u) < (1-\epsilon)c$. Here,  we call the subroutine $\text{{\sc REFILL}}(u)$ as described in Figure~\ref{fig:adjust:new}.
\item {\bf Case 1b:} $\text{{\sc \#Friends}}(u) \geq (1-\epsilon)c$. In this case, we do nothing and conclude the procedure.
\end{itemize}
\item {\bf Case 2:} $\text{{\sc Type}}(u) = \text{slack}$. In this case, we do nothing and conclude the procedure. 
\end{itemize}

\paragraph{Remark.} Note the main difference between the $\text{{\sc REFILL-NOW-FRIENDS}}(u)$ procedure in Section~\ref{sec:matching:sqrtn} (see Figure~\ref{fig:adjust:n}) and the $\text{{\sc REFILL}}(u)$ procedure in this section (see Figure~\ref{fig:adjust:new}). The former procedure (Figure~\ref{fig:adjust:n}) ensures that no node ever gets more than $\sqrt{n}$ friends.  On the other hand, in this section a node $x \notin \text{{\sc Friends}}[u]$ can become a friend of $u$ {\em no matter} how many friends $x$ already has.  Thus, a node $x$ can possibly get more than $c$ friends. A phase, however, lasts for no more than $\eps^2 c^2 / 2$ edge insertion/deletions. Using this fact, we will show that a node $x$ can have at most $(1+\eps) c$-many  friends at any point in time.

\begin{lemma}
\label{lm:deletion:new}
Suppose that an edge is deleted  from the  graph $G$ and we run the procedure described above.
\begin{enumerate}
\item This causes at most $O(\eps c)$ edge insertions into $\kappa(G)$ and at most one edge deletion from   $\kappa(G)$.
\item  The procedure runs in $O(c)$ time in the worst case.
\end{enumerate}
\end{lemma}

\begin{figure}[htbp]
\centerline{\framebox{
\begin{minipage}{5.5in}
\begin{tabbing}
{\sc If} the list $\text{{\sc neighbors}}(u)$ is empty, {\sc Then} \\
\ \ \ \ \ \ \ \  \  \=  Set $\text{{\sc Type}}(u) \leftarrow \text{slack}$. \\
\> {\sc  Return}. \\
  Let $x$ be the first node in the list $\text{{\sc neighbors}}(u)$. \\ 
  {\sc While} $\left(\text{{\sc \#Friends}}(u) < c\right)$ \\
  \> {\sc If} $x \notin \text{{\sc Friends}}(u)$, {\sc Then} \\
   \> \ \ \ \ \ \ \ \ \ \= Add $x$ to $\text{{\sc Friends}}(u)$ and $u$ to   $\text{{\sc Friends}}(x)$,  so that  the edge $(u,x)$  \\
   \> \> becomes  part of  the kernel   $\kappa(G) = (V,\kappa(E))$.  Increment each of the counters  \\
 \> \> $\text{{\sc \#Friends}}(u)$, $\text{{\sc \#Friends}}(v)$ by one unit.  Update the pointers \\
 \> \>  $\text{{\sc Pointer}}[u,x]$  and  $\text{{\sc Pointer}}[x,u]$.\\
 \>  {\sc If} $x$ is the last node in  $\text{{\sc neighbors}}(u)$, {\sc Then}  \\ 
 \> \> exit the {\sc While} loop. \\
  \>  {\sc Else} \\
\>    \> Let $y$ be the node that succeeds $x$ in  the list $\text{{\sc neighbors}}(u)$. \\
  \>  \>  Set $x \leftarrow y$. \\ 
 {\sc If} $\text{{\sc \#Friends}}(u) < c$, {\sc Then} \\
 \> Set $\text{{\sc Type}}(u) \leftarrow \text{slack}$.
  \end{tabbing}
\end{minipage}
}}
\caption{\label{fig:adjust:new} $\text{{\sc REFILL}}(u)$.} 
\end{figure}

\subsubsection{Proof of Theorem~\ref{th:phase:new}.}
\label{sec:phase:proof:new}

The first  part of the theorem immediately follows from Lemmas~\ref{lm:insertion:new},~\ref{lm:deletion:new}. We now focus on the second part.
Note that at most one edge is deleted from $\kappa(G)$ after an update (edge insertion or deletion) in $G$. Since the phase lasts for $\eps^2 c^2/2$ updates in $G$,  at most $\eps^2 c^2/2$ edge deletions occur in $\kappa(G)$ during the  phase. To complete the proof, we will show that  the corresponding number of edge insertions in $\kappa(G)$ is  also $O(\eps^2 c^2)$.

For an edge insertion in $G$, at most one edge is inserted into $\kappa(G)$. For an edge deletion in $G$, there can be $O(\eps c)$ edge insertions in $\kappa(G)$, {\em but only if} the subroutine $\text{REFILL}(.)$ is called. Lemma~\ref{lm:phase:new} shows that at most $\eps c$ calls are made to the subroutine $\text{REFILL}(.)$ during the phase. This implies that at most $O(\eps^2 c^2)$ edge insertions occur in $\kappa(G)$ during the phase, which proves the second part of the theorem.

\begin{lemma}
\label{lm:phase:new}
The subroutine $\text{{\sc REFILL}}(.)$ is called at most $\eps c$ times during the phase.
\end{lemma}

\begin{proof}
When the phase begins, every tight node has exactly $c$ friends (see Theorem~\ref{th:preprocess}), and  during the phase,  the status of a  node $u$ is changed from slack to tight only when $|\kappa(\N_u)|$  becomes at least $c$. On the other hand,  the subroutine $\text{{\sc REFILL}}(u)$ is called only when $u$ is tight and $|\kappa(\N_u)|$ falls below $(1-\eps)c$. Thus, each call to   $\text{{\sc REFILL}}(.)$  corresponds to a scenario where a tight node has lost at least  $\eps c$  friends. Each  edge deletion in $G$ leads to  at most  two such losses (one for each of the endpoints), whereas an edge insertion in $G$ leads to no such event.  Since the phase  lasts for $\eps^2 c^2/2$ edge insertions/deletions in $G$, a counting argument shows that it can lead to at most $(\eps^2 c^2/2) \cdot 2 / (\eps c)  = \eps c$ calls to $\text{{\sc REFILL}}(.)$.
\end{proof}

It remains to prove the final part of the theorem, which states that the algorithm maintains an $(\eps,c)$-kernel. Specifically, we will show that  throughout the duration of the phase, the subgraph $\kappa(G)$ maintained by the algorithm  satisfies Invariants~\ref{inv:1}--\ref{inv:3}.

\begin{lemma}
\label{lm:insertion:1:new}
Suppose that  $\kappa(G)$ satisfies Invariants~\ref{inv:2},~\ref{inv:3} before an edge update in  $G$.
Then these  invariants  continue to hold even after we modify $\kappa(G)$ as per the procedure in Section~\ref{sec:insertion:new} (resp.~Section~\ref{sec:deletion:new}).
\end{lemma}

\begin{proof}
Follows from the descriptions of the procedures in Sections~\ref{sec:insertion:new} and~\ref{sec:deletion:new}.
\end{proof}

Recall that in the beginning of the phase, the graph $\kappa(G)$ is a $(0,c)$-kernel of $G$. Since every $(0,c)$-kernel is also an $(\eps,c)$-kernel, we repeatedly invoke Lemma~\ref{lm:insertion:1:new} after each update in $G$, and conclude that $\kappa(G)$ satisfies Invariants~\ref{inv:2},~\ref{inv:3} throughout the duration of the phase. For the rest of this section, we focus on proving the remaining Invariant~\ref{inv:1}.

Fix any node $v \in V$. When the phase begins, the subgraph $\kappa(G)$ is a $(0,c)$-kernel of $G$, so that we have $|\kappa(\N_v)| \leq c$. During the phase, the node $v$ can get new friends under two possible situations.
\begin{enumerate}
\item An edge incident to $v$ has just been inserted into (resp.~deleted from) the graph $G$, and the procedure in Section~\ref{sec:insertion:new} (resp.~the subroutine $\text{REFILL}(v)$ in Figure~\ref{fig:adjust:new}) is going to be called.
\item The subroutine $\text{REFILL}(u)$ is going to be called for some $u \in \N_v$.
\end{enumerate}

If we are  in situation (1)  and the node $v$ already has at least $c$ friends, then the procedure under consideration will not end up adding any  more node to $\kappa(\N_v)$. Thus, it suffices to show that the node $v$ can get at most $\eps c$ new friends during the phase under  situation (2). Note that each call to $\text{{\sc REFILL}}(u)$, $u \neq v$,  creates at most one new friend for $v$. Accordingly, if we  show that the subroutine $\text{{\sc REFILL}}(.)$ is called at most $\eps c$ times during the entire phase, then this will suffice to conclude the proof of Theorem~\ref{th:phase:new}. But this has already been done in Lemma~\ref{lm:phase:new}.


\section{$(4+\eps)$-approximate matching in $O(m^{1/3}/\eps^2)$ worst-case update time}
\label{sec:4matching}
\label{app:matching:3}

Fix any $\eps \in (0,1/6)$. We present an algorithm for maintaining a $(4+\eps)$-approximate matching $M$ in a dynamic graph $G = (V,E)$ with $O(m^{1/3}/\eps^2)$ worst-case update time.

\subsection{Overview of our approach.}
\label{sec:4matching:overview}

As in Section~\ref{sec:3matching:overview}, we partition the sequence of updates (edge insertions/deletions) in $G$ into {\em phases}. Each phase lasts for $\eps^2 c^2/2$ consecutive updates in $G$.  Let $G_{i,t}$ denote the state of  $G$ just after the $t^{th}$ update in phase $i$, where $0 \leq t \leq \eps^2 c^2/2$. The initial state of the graph, before it starts changing, is given by $G_{1,0}$. Thus, we reach the graph $G_{i,t}$ from $G_{1,0}$ after a sequence of $(i-1) \cdot (\eps^2 c^2/2) + t$ updates in $G$.

For the rest of this section, we  focus on describing our algorithm for any given phase $i \geq 1$. We define $m \leftarrow |E_{i,0}|$ to be the number of edges in the input graph in the beginning of the phase, and  set $c = m^{1/3}$. Since the phase lasts for only $O(\eps^2 m^{2/3})$ updates in $G$, it follows that $|E_{i,t}| = O(m)$ for all $0 \le t \le \epsilon^2 c^2/2$.

\paragraph{Preprocessing.}
If $i = 1$, then we build a $(0,c)$-kernel $\kappa(G_{1,0})$  in the beginning of the phase as per Theorem~\ref{th:preprocess}.  Next, we build a maximal matching $M_{1,0}$ in $\kappa(G_{1,0})$. 
Overall, this takes $O(|V| + |E_{1,0}|)$ time. \\

\paragraph{Algorithm for each phase.}
We first present the high level approach. The main difference between the framework of Section~\ref{sec:3matching:overview} and that of this section is as follows. In Section~\ref{sec:3matching:overview}, the algorithm  rebuilds the kernel and a corresponding matching from scratch at the beginning of each phase, and we amortize
the cost of the rebuild over the
phase. In this section, to achieve a worst-case running time we do this rebuilding ``in the background'' during the phase.
This means that at the beginning of a phase we start with an empty graph $G^*$ and an empty kernel and insert $O(m/(\eps^2 c^2))$ edges into 
$G^*$ and its kernel during each update in $G$. Thus, edge insertions into $G^*$ need to be handled in constant time.
We can handle ``bunch updates'' into $G^*$ to build the kernel efficiently, but we cannot efficiently
update a $3/2$-approximate matching in the kernel as each edge insertion takes $O(c)$ time  (Theorem~\ref{th:maxdeg:2}). However, we can update a {\em maximal} matching, i.e., a
2-approximate matching in constant time per
edge insertions (Theorem~\ref{th:maxdeg:1}). Thus, we  run the 2-approximation algorithm on the kernel instead of the $3/2$-approximation algorithm, leading to 
a $(4+\eps)$ overall approximation. Additionally, we  need to perform the updates of the current phase in $G^*$. For that we use basically 
the same algorithm as in Section~\ref{sec:3matching:maintain:1}. As a result at the end of a phase $G^* = G$ and the kernel that we built for $G^*$ is only an
$(\eps,c)$ kernel of $G$, i.e., we have to start the next phase with a $(\eps,c)$-kernel instead of a $(0,c)$-kernel. This, however, degrades the approximation ratio
only by an additional factor of $\eps$.

To summarize our algorithm for phase $i$  has two  {\em  components}.
\begin{itemize} 
\item {\em Dealing with the current phase.} Just before the start of phase $i$, an $(\eps,c)$-kernel  $\kappa(G_{i,0})$ and a maximal matching $M_{i,0}$ in $\kappa(G_{i,0})$ are handed over to us. Then, as $G$ keeps changing, we keep modifying the subgraph $\kappa(G)$ and the matching $M$. Till the end of phase $i$, we ensure that  $\kappa(G)$ remains a $(2\eps,c)$-kernel of $G$ and that  $M$ remains a maximal matching in $\kappa(G)$. Hence, by Theorem~\ref{th:invariant}, the matching $M$ gives a $(4+ \eps)$-approximation\footnote{To be precise, the theorem shows a $(4+12\eps)$-approximation but running the algorithm with
$\eps' = \eps/12$ results in a $(4 + \eps)$-approximation.} to the maximum matching in $G$ throughout the duration of phase $i$.
\item {\em Preparing for the next phase.} We build a new $(\eps,c)$-kernel (and a maximal matching in it) for $G_{i+1,0}$ in the background. They are handed over to  the algorithm for phase $(i+1)$   at the start of phase $(i+1)$. 
\end{itemize}

\noindent We  elaborate upon each of these components in more details in Section~\ref{sec:kernel} (see Lemmas~\ref{lm:main:update},~\ref{lm:rebuild}). For maintaining a kernel in a dynamic graph, both these components use a procedure that is essentially the one described in Section~\ref{sec:3matching:maintain:1}. 
However, a more fine tuned analysis of the procedure becomes necessary.

Finally, Theorem~\ref{th:invariant} and Lemmas~\ref{lm:main:update},~\ref{lm:rebuild} give us the desired guarantee.

\begin{theorem}
\label{th:worstcase}
We can maintain a $(4+\eps)$-approximate matching in a dynamic graph $G = (V,E)$ in $O(m^{1/3}/\eps^2)$ worst-case update time.
\end{theorem}

\subsection{Algorithm for maintaining and building the kernel during each phase.}
\label{sec:kernel}

In Subsections~\ref{sec:insertion} and~\ref{sec:deletion}
we  present an algorithm for maintaining a kernel in a fully dynamic graph. 
It is basically a relaxation of the algorithm from Subsection~\ref{sec:3matching:maintain:1}
and is based on the concept of an ``epoch''.

\begin{definition}
\label{def:epoch}
For  $c \geq 1$, $\lambda \in  (0,1)$,  a sequence of updates in  a graph is called a {\em $(\lambda,c)$-epoch} iff it contains at most $\lambda^2 c^2/2$ edge deletions (which can be arbitrarily  interspersed with any number of edge insertions).
\end{definition}

The resulting algorithm is summarized in the theorem below.

\begin{theorem}
\label{th:epoch}
Fix any  $c \geq 1$ and  $\eps, \lambda > 0$ with $(\eps +\lambda) < 1/3$. Consider any $(\lambda,c)$-epoch in the input graph $G = (V,E)$, and suppose that we are given an $(\epsilon,c)$-kernel $\kappa(G)$ in the beginning of the epoch. Then we have an algorithm for updating $\kappa(G)$ after each update in $G$. The algorithm satisfies three properties.
\begin{enumerate}
\item An edge insertion in $G$ is handled in $O(1)$ worst-case  time, and this leads to at most one edge insertion in $\kappa(G)$ and zero edge deletion in $\kappa(G)$. 
\item An edge deletion in $G$ is handled in $O(c)$ worst-case  time, and this leads to at most $O(\lambda c)$ edge insertions in $\kappa(G)$ and at most one edge deletion in $\kappa(G)$. 
\item Throughout the duration of the epoch,  $\kappa(G)$ remains a $(\lambda+\eps,c)$-kernel of $G$.
\end{enumerate}
\end{theorem}


We use the kernel update algorithm in the two components of a phase, which we describe next in detail.
Recall  that  $G_{i,t}$ denotes the state of  $G$ just after the $t^{th}$ update in  phase $i \geq 1$ for all $0 \leq t \leq \epsilon^2 c^2/2$, $m = |E_{i,0}|$, and $c = m^{1/3}$. 

\subsubsection{Dealing with the current phase}
\label{sec:4matching:current}

In the beginning of  phase $i$, we are given an $(\eps,c)$-kernel $\kappa(G_{i,0})$ and a maximal matching $M_{i,0}$ in $\kappa(G_{i,0})$. Since a phase lasts for $\eps^2 c^2/2$ edge updates in $G$, we  treat phase $i$  as an $(\eps,c)$-epoch (see Definition~\ref{def:epoch}). Thus, after each edge insertion/deletion in $G$ during phase $i$, we modify $\kappa(G)$  as per  Theorem~\ref{th:epoch}. This ensures that $\kappa(G)$ remains an $(2\eps,c)$-kernel till the end of phase $i$. After each update in $\kappa(G)$, we modify the matching $M$  using Theorem~\ref{th:maxdeg:1} so as to ensure that $M$ remains a maximal matching in $\kappa(G)$. \\

\noindent We  show that this procedure requires $O(c)$ time in the worst case to handle an  update in $G$.
\begin{itemize}
 \item {\em Case 1: An edge is inserted into the graph $G$}. By Theorem~\ref{th:epoch}, in this case updating the kernel $\kappa(G)$ requires $O(c)$ time, and this results in at most one edge insertion into $\kappa(G)$ and zero edge deletion from $\kappa(G)$. By Theorem~\ref{th:maxdeg:1}, updating the matching $M$ requires $O(1)$ time.
\item {\em Case 2: An edge is deleted from the graph $G$.} By Theorem~\ref{th:epoch}, in this case updating the kernel $\kappa(G)$ requires $O(c)$ time, and this results in $O(\eps c)$ edge insertions into $\kappa(G)$ and at most one edge deletion from $\kappa(G)$. Since the maximum degree of a node in $\kappa(G)$ is $O(c)$, updating the matching $M$ requires $O(\eps c) + O(c) = O(c)$ time (see Theorem~\ref{th:maxdeg:1}).
\end{itemize}

\begin{lemma}
\label{lm:main:update}
Suppose that when phase $i$ begins, we are given an $(\eps,c)$-kernel $\kappa(G)$ of $G$ and a maximal matching $M$ in $\kappa(G)$. After each update in $G$ during phase $i$, we can modify  $\kappa(G)$ and $M$ in $O(c)$ time. Till the end of phase $i$,  $\kappa(G)$ remains an $(2\eps,c)$-kernel of $G$ and $M$ remains a maximal matching in $\kappa(G)$.
\end{lemma}

\subsubsection{Preparing   for the next phase.}
\label{sec:4matching:next}


Before phase $i$ begins, in $O(1)$ time  we initialize  a graph $G^* = (V, E^*)$ with $E^* = \emptyset$, a $(0,c)$-kernel $\kappa(G^*) = (V, \kappa(E^*))$, and a maximal matching $M^* = \emptyset$ in $\kappa(G^*)$. To achieve constant time we use an uninitialized array for the $\text{{\sc Type}}$-bit and  assume that every node
for which the bit is not initialized is of type slack.\\

\begin{figure}[htbp]
\centerline{\framebox{
\begin{minipage}{5.5in}
\begin{tabbing}
{\sc If} an edge $(u,v)$ is inserted into $G$, {\sc Then} \\
\ \ \ \ \ \ \ \ \ \= $E^* \leftarrow E^* \cup \{(u,v)\}$. \\
{\sc Else if}  an edge $(u,v)$ is deleted from $G$,  {\sc Then} \\
\> $E^* \leftarrow E^* \setminus \{(u,v)\}$. \\
{\sc For} $i = 1 $ to $2m/(\eps^2 c^2)$ \\
\> {\sc If}  $E \setminus E^* \neq \emptyset$, {\sc Then} \\
\> Pick any edge $e \in E \setminus E^*$. \\
\>  $E^* \leftarrow E^* \cup \{e\}$.
\end{tabbing}
\end{minipage}
}}
\caption{\label{fig:updateg} $\text{{\sc Update-}}G^*$.} 
\end{figure}

\noindent  {\em After each update in $G$ in phase $i$, we call  $\text{{\sc Update-}}G^*$ (Figure~\ref{fig:updateg}). This ensures the following properties.}
\begin{itemize}
\item (P1) Each edge insertion in $G$ leads to at most $O(m/(\eps^2 c^2))$ edge insertions in $G^*$. 
\item (P2) Each edge deletion in $G$ leads to at most one edge deletion in $G^*$ and at most $O(m/(\eps^2 c^2))$ edge insertions in $G^*$.
\item (P3) After each  edge insertion/deletion in $G$, changing the graph $G^*$ requires $O(m/(\eps^2 c^2))$ time.
\item (P4) We always have $E^* \subseteq E$. Furthermore, at the end of phase $i$, we have $E^* = E$. This  holds since $m$ is the size of $E$ in the beginning of phase $i$, and since the phase lasts for $\eps^2 c^2/2$ updates in $G$.
\end{itemize}

\noindent Now, consider the sequence of updates in $G^*$ that take place during phase $i$.  (P1) and (P2) guarantee that this sequence contains at most $\eps^2 c^2/2$ edge deletions. So this sequence can be treated as an $(\eps,c)$-epoch in $G^*$ (see Definition~\ref{def:epoch}).  Thus, after each update in $G^*$, we modify the subgraph $\kappa(G^*)$ using Theorem~\ref{th:epoch}, and subsequently, we modify the maximal matching $M^*$ in $\kappa(G^*)$ using Theorem~\ref{th:maxdeg:1}. \\

\noindent Since  $\kappa(G^*)$ was a $(0,c)$-kernel of $G^*$ in the beginning of the epoch, we get the following guarantee.

\begin{itemize}
\item (P5) Throughout the duration of phase $i$, the graph $\kappa(G^*)$ is an $(\eps,c)$-kernel of $G^* = (V,E^*)$ and $M^*$ is a maximal matching in $\kappa(G^*)$.
\end{itemize}

\noindent  By (P1) and (P2), an update  in $G$ leads  to $O(m/(\eps^2 c^2))$ edge insertions and at most one edge deletion in $G^*$. Each of these edge insertions in $G^*$ further leads to at most one edge insertion in $\kappa(G^*)$, while the potential edge deletion in $G^*$ leads to at most one edge deletion  and $O(\eps c)$ edge insertions in $\kappa(G^*)$ (see Theorem~\ref{th:epoch}). To summarize, an update in $G$ leads to the following updates in $G^*$ and $\kappa(G^*)$.
\begin{itemize}
\item At most one edge deletion and $O(m/(\eps^2 c^2))$ edge insertions in $G^*$.
\item At most one edge deletion and $O(\eps c+m/(\eps^2 c^2))$ edge insertions in $\kappa(G^*)$.
\end{itemize} 
Thus, by Theorems~\ref{th:epoch} and~\ref{th:maxdeg:1}, we get  the following bound on the update time

\begin{itemize}
\item (P6) After each edge update in $G$, the graph $\kappa(G^*)$ and the maximal matching $M^*$ in $\kappa(G^*)$ can be modified in $O(c) + O(\eps c + m/(\eps^2 c^2)) = O(c+m/(\eps^2 c^2))$ time.
\end{itemize}

\noindent Using (P3) and (P6), we derive that after each update in $G$, the  time required to modify the structures $G^*, \kappa(G^*)$ and $M^*$  is $O(c+m/(\eps^2 c^2)) = O(m^{1/3}/\eps^2)$ in the worst case. By (P4) and (P5), the graph $G^*$ becomes identical with $G$ at the end of phase $i$, and at this point $\kappa(G^*)$ becomes an $(\eps,c)$-kernel of $G_{i+1,0}$. 

\begin{lemma}
\label{lm:rebuild}
Starting from the beginning of phase $i$, we can run an algorithm with the following properties.
\begin{itemize}
\item After each update in $G$ during phase $i$, it performs $O(m^{1/3}/\eps^2)$ units of computation.
\item At the end of phase $i$, it returns an $(\eps,c)$-kernel of $G_{i+1,0}$ and a maximal matching in this kernel.
\end{itemize}
\end{lemma}

\noindent Theorem~\ref{th:worstcase} now follows from Theorem~\ref{th:invariant} and Lemmas~\ref{lm:main:update},~\ref{lm:rebuild}.

\subsubsection{Handling an edge insertion into $G$ during the epoch.}
\label{sec:insertion}

Suppose that an edge $(u,v)$ is inserted into the graph $G = (V, E)$. To handle this edge insertion, we first  update the lists $\text{{\sc neighbors}}(u)$ and $\text{{\sc neighbors}}(v)$, and then process the edge as follows.

\begin{itemize}
\item {\bf Case 1:} Either $\text{{\sc Type}}(u)  = \text{tight}$ or $\text{{\sc Type}}(v)  = \text{tight}$.

 We do nothing and conclude the procedure.
\item {\bf Case 2:} Both $\text{{\sc Type}}(u) = \text{slack}$ and $\text{{\sc Type}}(v) = \text{slack}$.
\begin{itemize}
\item {\bf Case 2a:} Either $\text{{\sc \#Friends}}(u) \geq (1-\eps)c$ or $\text{{\sc \#Friends}}(v) \geq (1-\eps)c$.

If $\text{{\sc \#Friends}}(u) \geq (1-\eps)c$, then we set $\text{{\sc Type}}(u) \leftarrow \text{tight}$. Next, if $\text{{\sc \#Friends}}(v) \geq (1-\eps)c$, then we set $\text{{\sc Type}}(v) \leftarrow \text{tight}$.
\item {\bf Case 2b:} Both $\text{{\sc \#Friends}}(u) < (1-\eps)c$ and $\text{{\sc \#Friends}}(v) < (1-\eps)c$. 

We add the edge $(u,v)$ to the kernel $\kappa(G)$. Specifically, we  add  $u$ to the list $\text{{\sc Friends}}(v)$ and $v$ to the list  $\text{{\sc Friends}}(u)$,  update the pointers $\text{{\sc Pointer}}[u,v]$ and $\text{{\sc Pointer}}[v,u]$ accordingly, and increment each of  the counters $\text{{\sc \#Friends}}(u)$, $\text{{\sc \#Friends}}(v)$ by one unit.
\end{itemize}
\end{itemize}

\begin{lemma}
\label{lm:insertion}
Suppose that an edge is inserted into the  graph $G$ and we run the procedure described above.
\begin{enumerate}
\item This causes at most one edge insertion into $\kappa(G)$ and at most one edge deletion from   $\kappa(G)$.
\item  The procedure runs in $O(1)$ time in the worst-case.
\end{enumerate}
\end{lemma}

\subsubsection{Handling an edge deletion in $G$ during the epoch.}
\label{sec:deletion}

Suppose that an edge $(u,v)$ gets deleted from the graph $G = (V, E)$. To handle this edge deletion, we first consider the adjacency-list data structure for $G$, and update the lists $\text{{\sc neighbors}}(u)$ and $\text{{\sc neighbors}}(v)$.    

We  check if the edge $(u,v)$ was part of the kernel  $\kappa(G)$, and, if yes, then we delete $(u,v)$ from $\kappa(G)$.  Specifically,   we delete $u$ from $\text{{\sc Friends}}(v)$ and  $v$ from $\text{{\sc Friends}}(u)$ (using the pointers $\text{{\sc Pointer}}[u,v]$, $\text{{\sc Pointer}}[v,u]$), and  decrement each of the counters $\text{{\sc \#Friends}}(u)$, $\text{{\sc \#Friends}}(v)$ by one unit.

We then process the nodes $u$ and $v$ one after another. Below, we describe  only the procedure that runs on the node $u$. The procedure for the node $v$ is exactly the same.

\begin{itemize}
\item {\bf Case 1:} $\text{{\sc Type}}(u) = \text{tight}$. Here, we  check if the number of friends of $u$ has dropped below the prescribed limit  due to the edge deletion, and, accordingly, we  consider two possible sub-cases.
\begin{itemize}
\item {\bf Case 1a:} $\text{{\sc \#Friends}}(u) < (1-\lambda - \epsilon)c$. Here,  we call the subroutine $\text{{\sc REFILL-NOW}}(u)$ (see Figure~\ref{fig:adjust}).
\item {\bf Case 1b:} $\text{{\sc \#Friends}}(u) \geq (1-\lambda -\epsilon)c$. In this case, we do nothing and conclude the procedure.
\end{itemize}
\item {\bf Case 2:} $\text{{\sc Type}}(u) = \text{slack}$. In this case, we do nothing and conclude the procedure. 
\end{itemize}

\paragraph{Remark.} The procedure $\text{{\sc REFILL-NOW}}(u)$ (see Figure~\ref{fig:adjust}) is identical with the procedure $\text{{\sc REFILL}}(u)$ (see Figure~\ref{fig:adjust:new}) in Section~\ref{sec:3matching}, except that it uses a different threshold for the number of friends of a node.

\begin{lemma}
\label{lm:deletion}
Suppose that an edge is deleted  from the  graph $G$ and we run the procedure described above.
\begin{enumerate}
\item This causes at most $O(\lambda c)$ edge insertions into $\kappa(G)$ and at most one edge deletion from   $\kappa(G)$.
\item  The procedure runs in $O(c)$ time in the worst-case.
\end{enumerate}
\end{lemma}

\begin{figure}[htbp]
\centerline{\framebox{
\begin{minipage}{5.5in}
\begin{tabbing}
{\sc If} the list $\text{{\sc neighbors}}(u)$ is empty, {\sc Then} \\
\ \ \ \ \ \ \ \ \ \  \=  Set $\text{{\sc Type}}(u) \leftarrow \text{slack}$. \\
\> {\sc  Return}. \\
  Let $x$ be the first node in the list $\text{{\sc neighbors}}(u)$. \\ 
  {\sc While} $\left(\text{{\sc \#Friends}}(u) < (1-\eps)c\right)$ \\
  \> {\sc If} $x \notin \text{{\sc Friends}}(u)$, {\sc Then} \\
   \> \ \ \ \ \ \ \ \ \ \  \ \= Add $x$ to $\text{{\sc Friends}}(u)$ and $u$ to $\text{{\sc Friends}}(x)$,   so that  the edge $(u,x)$ becomes part of the \\
 \> \> kernel $\kappa(G) = (V,\kappa(E))$.  Increment each of the counters  $\text{{\sc \#Friends}}(u)$, $\text{{\sc \#Friends}}(x)$ \\
 \> \> by one unit.  Update the pointers  $\text{{\sc Pointer}}[u,x]$, $\text{{\sc Pointer}}[x,u]$. \\
 \>  {\sc If} $x$ is the last node in  $\text{{\sc neighbors}}(u)$, {\sc Then}  \\ 
 \> \> exit the {\sc While} loop. \\
  \>  {\sc Else} \\
\>    \> Let $y$ be the node that succeeds $x$  in the list $\text{{\sc neighbors}}(u)$. \\
  \>  \>  Set $x \leftarrow y$. \\ 
 {\sc If} $\text{{\sc \#Friends}}(u) < (1-\eps)c$, {\sc Then} \\
 \> Set $\text{{\sc Type}}(u) \leftarrow \text{slack}$.
  \end{tabbing}
\end{minipage}
}}
\caption{\label{fig:adjust} $\text{{\sc REFILL-NOW}}(u)$.} 
\end{figure}

\subsubsection{Proof of Theorem~\ref{th:epoch}}
\label{sec:epoch:proof}

The first and the second parts of the theorem follows from Lemma~\ref{lm:insertion} and Lemma~\ref{lm:deletion} respectively. It remains to show that the algorithm maintains a $(\lambda+\eps,c)$-kernel. Specifically, we will show that  throughout the duration of the $(\lambda,c)$-epoch, the subgraph $\kappa(G)$ maintained by the algorithm continues to satisfy Invariants~\ref{inv:1}--\ref{inv:3} (replacing $\eps$ by $\lambda+\eps$ in Invariants~\ref{inv:1},~\ref{inv:2}).

\begin{lemma}
\label{lm:insertion:1}
Suppose that  $\kappa(G)$ satisfies two conditions before an edge insertion (resp.~deletion) in  $G$.
\begin{enumerate}
\item $|\kappa(\N_v)| \geq (1-\lambda-\eps)c$ for all nodes $v \in \kappa_T(V)$.
\item For all nodes $u, v \in \kappa_S(V)$, if $(u,v) \in E$, then either $u \in \kappa(\N_v)$ or $v \in \kappa(\N_u)$.
\end{enumerate}
\noindent Then these two conditions  continue to hold even after we modify $\kappa(G)$ using the procedure in Section~\ref{sec:insertion} (resp.~Section~\ref{sec:deletion}).
\end{lemma}

\begin{proof}
Follows from the descriptions of the procedures in Sections~\ref{sec:insertion} and~\ref{sec:deletion}.
\end{proof}

Recall that before the $(\lambda,c)$-epoch begins, the graph $\kappa(G)$ is an $(\eps,c)$-kernel of $G$. Since every $(\eps,c)$-kernel is also a $(\lambda+\eps,c)$-kernel, we repeatedly invoke Lemma~\ref{lm:insertion:1} after each update in $G$, and conclude that $\kappa(G)$ satisfies Invariants~\ref{inv:2},~\ref{inv:3} (replacing $\eps$ by $\lambda+\eps$) throughout the duration of the epoch. For the rest of this section, we focus on proving the remaining Invariant~\ref{inv:1}. Specifically, we will show that at any point in time during the epoch, we have $|\kappa(\N_v)| \leq (1+\lambda+\eps)c$ for all $v \in V$.

Fix any node $v \in V$. When the $(\lambda,c)$-epoch begins, the subgraph $\kappa(G)$ is an $(\epsilon,c)$-kernel of $G$, so that we have $|\kappa(\N_v)| \leq (1+\eps)c$. During the epoch, the node $v$ can get new friends under two possible situations.
\begin{enumerate}
\item An edge incident to $v$ has just been inserted into (resp.~deleted from) the graph $G$, and the procedure in Section~\ref{sec:insertion} (resp.~the subroutine $\text{REFILL-NOW}(v)$) is going to be called.
\item The subroutine $\text{REFILL-NOW}(u)$ is going to be called for some $u \in \N_v$.
\end{enumerate}

If we are  in situation (1)  and the node $v$ already has more than $(1+\eps)c$ friends, then the procedure under consideration will not end up adding any  more node to $\kappa(\N_v)$. Thus, it suffices to show that the node $v$ can get at most $\lambda c$ new friends during the epoch under  situation (2). Note that each call to $\text{{\sc REFILL-NOW}}(u)$, $u \neq v$,  creates at most one new friend for $v$. Accordingly, if we  show that the subroutine $\text{{\sc REFILL-NOW}}(.)$ is called at most $\lambda c$ times during the entire epoch, then this will suffice to conclude the proof of Theorem~\ref{th:epoch}.

\begin{lemma}
\label{lm:epoch}
The subroutine $\text{{\sc REFILL-NOW}}(.)$ is called at most $\lambda c$ times during a $(\lambda,c)$-epoch.
\end{lemma}

\begin{proof}
When the epoch begins, every tight node has at least $(1-\eps)c$ friends, and  during the epoch,  the status of a  node $u$ is changed from slack to tight only when $|\kappa(\N_u)|$  exceeds $(1-\eps)c$. On the other hand,  the subroutine $\text{{\sc REFILL-NOW}}(u)$ is called only when $u$ is tight and $|\kappa(\N_u)|$ falls below $(1-\lambda-\eps)c$. Thus, each call to   $\text{{\sc REFILL-NOW}}(.)$  corresponds to a scenario where a tight node has lost at least  $\lambda c$  friends. Each  edge deletion in $G$ leads to  at most  two such losses (one for each of the endpoints), whereas an edge insertion in $G$ leads to no such event.  Since a $(\lambda,c)$-epoch  contains at most $\lambda^2 c^2/2$ edge deletions (Definition~\ref{def:epoch}), a counting argument shows that such an epoch can result in at most $(\lambda^2 c^2/2) \cdot 2 / (\lambda c)  = \lambda c$ calls to $\text{{\sc REFILL-NOW}}(.)$.
\end{proof}

\bibliographystyle{plain}
\bibliography{citations}

\end{document}